\documentclass{article}

\usepackage[preprint,nonatbib]{neurips_2022}

\usepackage[utf8]{inputenc} 
\usepackage[T1]{fontenc}    
\usepackage{hyperref}       
\usepackage{url}            
\usepackage{booktabs}       
\usepackage{amsfonts}       
\usepackage{nicefrac}       
\usepackage{microtype}      
\usepackage{xcolor}         

\usepackage{pifont} 
\usepackage{amsmath} 
\usepackage{rotating} 
\usepackage{adjustbox} 
\usepackage{amsthm} 
\usepackage{textcomp} 
\usepackage{amsfonts} 
\usepackage{bm} 
\usepackage{amsthm} 
\usepackage{mathrsfs} 
\usepackage{thmtools, thm-restate} 
\usepackage{graphicx} 
\usepackage{subfigure} 
\usepackage{chngcntr} 
\usepackage{wrapfig} 
\usepackage{hyperref} 
\usepackage{multirow} 
\usepackage{makecell} 
\usepackage{caption} 
\usepackage[T1]{fontenc} 
\usepackage{mfirstuc}  
\usepackage[capitalize,noabbrev]{cleveref}
\usepackage[algo2e,linesnumbered,noend,ruled]{algorithm2e}  
\SetAlFnt{\footnotesize}
\SetAlCapFnt{\footnotesize}
\SetAlCapNameFnt{\footnotesize}
\usepackage{bbm}  
\usepackage[most]{tcolorbox}
\tcbset{
    frame code={}
    center title,
    left skip=0pt,
    left=0pt,
    leftrule=-5pt,
    rightrule=-5pt,
    right=0pt,
    top=0pt,
    bottom=0pt,
    colback=black!4,
    width=1\textwidth,
    boxsep=5pt,
    arc=0pt,outer arc=0pt,
    }

\theoremstyle{plain}

\theoremstyle{definition}

\newtheorem{fact}{Fact}
\newtheorem{remark}{Remark}

\newcommand{\thicktilde}[1]{\mathbf{\tilde{\text{$#1$}}}}

\newcommand{\AsyncActionFirstCap}{Off-beat~}

\newcommand{\AsyncAction}{off-beat~}
\newcommand{\AsyncActionNotilde}{off-beat}
\newcommand{\AsyncActionUpperCase}{Off-Beat~}
\newcommand{\ie}{\textit{i.e.}}
\newcommand{\eg}{\textit{e.g.}}
\newcommand{\MAASYNCENVS}{OBMAS}
\newcommand{\OurMethod}{LeGEM}
\newcommand{\OurMethodTilde}{LeGEM~}
\newcommand {\thmbox}[2]{\begin{tcolorbox}\begin{#1} #2 \end{#1}\end{tcolorbox}}

\newcommand*{\rom}[1]{\expandafter\@slowromancap\romannumeral #1@}

\title{\AsyncActionUpperCase Multi-Agent Reinforcement Learning}

\author{%
  Wei Qiu\thanks{Correspondence to \texttt{qiuw0008@e.ntu.edu.sg}}~~\textsuperscript{\textnormal{a}}, 
  Weixun Wang\textsuperscript{\textnormal{b}}, 
  Rundong Wang\textsuperscript{\textnormal{a}}, 
  Bo An\textsuperscript{\textnormal{a}}, 
  Yujing Hu\textsuperscript{\textnormal{c}}, 
  Svetlana Obraztsova\textsuperscript{\textnormal{a}}, \\
  \textbf{Zinovi Rabinovich}\textsuperscript{\textnormal{a}}\textbf{,}
  \textbf{Jianye Hao}\textsuperscript{\textnormal{b}}\textbf{,}
  \textbf{Yingfeng Chen}\textsuperscript{\textnormal{c}}\textbf{,}
  \textbf{Changjie Fan}\textsuperscript{\textnormal{c}} \\
  \textsuperscript{a}Nanyang Technological University, Singapore\\
  \textsuperscript{b}Tianjin University, Tianjin, China\\
  \textsuperscript{c}Fuxi AI Lab, Netease, Inc., Hangzhou, China
}

\begin{document}

\maketitle

\begin{abstract}
We investigate model-free multi-agent reinforcement learning (MARL) in environments where \textit{\AsyncAction}actions are prevalent, \ie, all actions have pre-set execution durations. During execution durations, the environment changes are influenced by, but not synchronised with, action execution. Such a setting is ubiquitous in many real-world problems. However, most MARL methods assume actions are executed immediately after inference, which is often unrealistic and can lead to catastrophic failure for multi-agent coordination with \AsyncAction actions. In order to fill this gap, we develop an algorithmic framework for MARL with \AsyncAction actions. We then propose a novel episodic memory, \OurMethod, for model-free MARL algorithms. \OurMethod~builds agents' episodic memories by utilizing agents' individual experiences. It boosts multi-agent learning by addressing the challenging temporal credit assignment problem raised by the \AsyncAction actions via our novel reward redistribution scheme, alleviating the issue of non-Markovian reward. We evaluate \OurMethod~on various multi-agent scenarios with \AsyncAction actions, including Stag-Hunter Game, Quarry Game, Afforestation Game, and StarCraft II micromanagement tasks. Empirical results show that \OurMethod~significantly boosts multi-agent coordination and achieves leading performance and improved sample efficiency.
\end{abstract}

\section{Introduction}\label{sec:intro}



In Multi-Agent Reinforcement Learning (MARL), multiple agents act interactively and complete tasks in a sequential decision-making manner with Reinforcement Learning (RL). It has made remarkable advances in many domains, including autonomous systems~\cite{cao2012overview,huttenrauch2017guided,zhou2020smarts} and real-time strategy (RTS) video games~\cite{vinyals2019grandmaster}.
By the virtue of the \textit{centralised training with decentralised execution} (CTDE)~\cite{oliehoek2008optimal} paradigm, which aims to tackle the scalability and partial observability challenges in  MARL, many CTDE-based MARL methods are proposed~\cite{foerster2017counterfactual,sunehag2017value,rashid2018qmix,wang2019action,son2019qtran,wang2021dop,kuba2021trust,pan2021regularized}. With these methods, an agent executes actions only via feeding its individual observations independently and optimizes its own policy with access to global trajectories centrally.


Despite the recent successes of MARL, learning
effective multi-agent coordination policies for complex multi-agent systems remains challenging. One key challenge is the \textit{\AsyncActionNotilde} actions, \ie, all actions have pre-set execution durations\footnote{In the RL literature~\cite{ramstedt2019real,bouteiller2020reinforcement}, action execution durations are called \textit{delays of actions}. In this paper, we use the term \textit{execution durations}, which is self-consistent with \AsyncAction actions defined in Sec. \ref{sec:async_decpomdp}.} and during the execution durations, the environment changes are influenced by, but not synchronised with, action execution (an illustrative scenario is shown in Fig. \ref{fig:async_example}). However, Dec-POMDP~\cite{oliehoek2016concise}, which underpins many CTDE-based MARL methods, hinges on the assumption that actions are executed momentarily after inference, leading to catastrophic failure for \textit{centralized training} on various off-beat multi-agent scenarios (\MAASYNCENVS). To fill this gap, 
\begin{figure}[ht]
    \vspace{-0.4cm}
    \centering
    \includegraphics[scale=0.415]{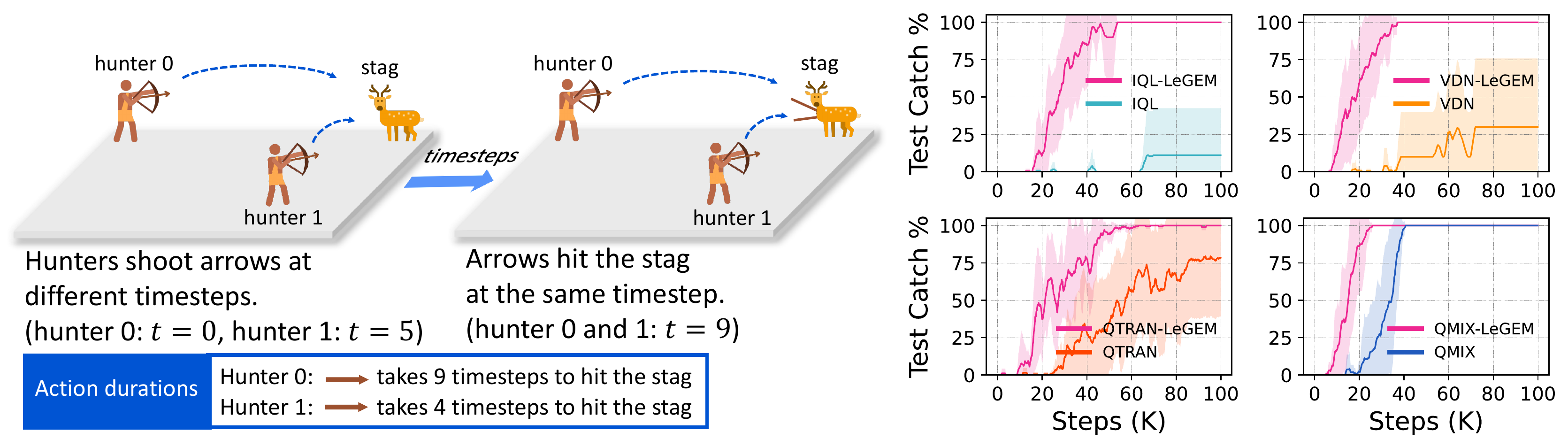}
    \vspace{-0.6cm}
    \caption{\footnotesize{\textbf{An illustrative scenario:} two-agent stag-hunter game, where two agents (hunters)  have only partial observations, different durations of the shoot action, and cannot communicate. The goal is to catch the stag and they are rewarded when their shot hits -- as in, completion of the action is synchronised, the stag at the same time. Both agents can see the stag. As the shoot action durations of the two agents are different, to catch the stag, the two agents should shoot the arrow at different timesteps given the distances. Though the scenario is easy for human beings, it is hard for MARL agents due to the action duration. \textbf{Experiment results:} 
    in this scenario, the optimal policy for agent $0$ is to shoot the arrow at timestep $0$ while the optimal policy for agent $1$ is to shoot the arrow at timestep $5$. Such asynchronous property of \MAASYNCENVS~ motivates agents to learn tacit policies. The curves show that VDN and IQL fail to learn coordination policies even in this simple scenario. 
    With \OurMethod, MARL methods gain enhanced performances as well as improved sample efficiency.} 
    }
    \vspace{-0.7cm}
    \label{fig:async_example}
\end{figure}
we study MARL in settings where \AsyncAction actions are prevalent. Such setting is very common in many real-world problems. For example, in the traffic light control problem, traffic lights in the conjunctions of the road network have pre-set execution time which is set asynchronously.


The problem of \AsyncAction actions in MARL has yet to be investigated and tackled.
Training MARL policies in \MAASYNCENVS~is challenging: (i) Each agent's actions can have a variety of execution durations, which augments the order of complexity of \MAASYNCENVS~during decentralized execution, resulting in failure of the coordination; (ii) The action durations are unknown to agents during individual executions, and communication is constrained and not always feasible, making it non-trivial to model the environment;  (iii) During training, both the temporal credit assignment with TD-learning~\cite{sutton2018reinforcement} and the \textit{inter-agent} credit assignment with value decomposition methods~\cite{rashid2018qmix} cannot perform well due to the displaced rewards in multi-agent replay.
With \AsyncAction actions, the nonstationarity issue, which mainly stems from rewards' time dependency on the agents' past actions, is exacerbated.

While actions durations are ubiquitous, existing works only focus on single-agent settings, \ie, delay, in RL. Many approaches~\cite{walsh2009learning,ramstedt2019real,xiao2019thinking} augment the state space with the queuing actions to be executed into the environment. However, such state-augmentation trick leads to exponentially increasing training samples with the growing action duration, making training intractable~\cite{derman2020acting}. Chen et al.~\cite{chen2020delay} extend the delayed MDP~\cite{ramstedt2019real} and propose Delayed Markov Game for MARL. However, on one hand, such state-augmentation treatment is confined to short delays, \eg,  one timestep delay; on the other hand, the delayed timestep of the actions is privileged information, which is not available in many scenarios. Recent works on macro-actions~\cite{xiao2020macro,xiao2020multi} introduce asynchronous actions by designing macro-actions with prior environment knowledge. Macro-actions are different from options in hierarchical RL (HRL)~\cite{sutton1999between,bacon2017option} in that the later is not manually designed but learned. The key difference between macro-actions and off-beat actions is that macro-actions are high-level actions while off-beat actions are primitive actions.  Unfortunately, the \textit{inter-agent} credit assignment is still a challenge of HRL in \MAASYNCENVS~and the asynchronous \footnote{We clarify the term \textit{asynchronous}: actions that simultaneously committed into the environment by all agents in MARL will not complete their respective action durations at the same time in future timesteps.} nature of \AsyncAction actions undermines the temporal credit assignment of \textit{centralized training}, causing poor sample efficiency and unsatisfactory performance (more discussions can be found in the related works sections in Sec.~\ref{sec:related_works}).



We aim to address the aforementioned issues. We first propose \AsyncAction Dec-POMDP. We then instantiate a new class of episodic memory, \OurMethod, for model-free MARL algorithms. \OurMethod~boosts multi-agent learning by addressing the challenging temporal credit assignment problem raised by the \AsyncAction actions via our novel levelled graph-based temporal recency reward redistribution scheme. Specifically, each agent maintains \OurMethod~and during centralized training, each agent searches the pivot timestep given observations from its graph. The pivot timestep is the timestep wherein the \AsyncAction reward relates to the given node. The pivot timesteps of each agent are ranked, in which the final pivot timestep will be chosen by recency and later used for reward redistribution and target estimation in TD-learning. We evaluate our method on Stag-Hunter Game, Quarry Game, Afforestation Game, and StarCraft II micromanagement tasks. Empirical results show that our method significantly boosts multi-agent coordination and achieves leading performance as well as improved sample efficiency.



\section{Preliminaries}\label{sec:preliminaries}
\textbf{Dec-POMDP.} A cooperative MARL problem can be modeled as a \textit{decentralised partially observable Markov decision process} (Dec-POMDP) which can be formulated as a tuple $\langle\mathcal{S},\mathcal{U}, \mathcal{P},R, O,\mathcal{N}, \gamma\rangle$, where $\bm{s} \in \mathcal{S}$ denotes the state of the environment. Each agent $i \in \mathcal{N} := \{1,...,N\} $ chooses an action $u_i \in \mathcal{U}$ at each timestep, forming a joint action vector, $\bm{u} := [u_i]_{i=1}^N \in \mathcal{U}^N$. The Markovian transition function can be defined as $\mathcal{P}(\bm{s}'|\bm{s},\bm{u}):\mathcal{S} \times \mathcal{U}^N \times \mathcal{S}\mapsto [0,1]$, transiting one state of current timestep to the state of next timestep conditioned on current state and joint action. Every agent shares the reward and the reward function is $R(\bm{s},\bm{u}): \mathcal{S} \times \mathcal{U}^N \mapsto \mathcal{R}$. $\gamma \in [0,1)$ is the discount factor. Due to \textit{partial observability}, each agent has individual partial observation $o \in \mathcal{O}$, according to the observation function $O(\bm{s},i): \mathcal{S} \times \mathcal{N} \mapsto \mathcal{O}.$ The goal of each agent is to optimize its own policy $\pi_i(u_i|\tau_i) : \mathcal{T} \times \mathcal{U} \mapsto [0,1]$ given its action-observation-reward history $\tau_i \in \mathcal{T} := (\mathcal{O} \times \mathcal{U})$.

\textbf{Multi-Agent Reinforcement Learning.} MARL aims to learn optimal policies for all the agents in the team. With TD-learning and a global Q value proxy $Q^{\operatorname{tot}}$ for the optimal $Q^{*}$, $\{Q_i\}^{N}_{i=1}$ are optimized via minimizing the  loss~\cite{watkins1992q,mnih2015human}: $\theta^{*} = \arg\min_{\theta^{*}} \mathcal{L} (\theta) :=\mathbb{E}_{D^{\prime} \sim \mathcal{D}}[(y_{t}^{\text {tot}}-Q_{\theta}^{\text {tot }}(\boldsymbol{s}_{t}, \boldsymbol{u}_{t}))^{2}]$, where $y_{t}^{\mathrm{tot}}=r_{t}+\gamma \max _{\boldsymbol{u}^{\prime}} Q_{\bar{\theta}}^{\mathrm{tot}}(\boldsymbol{s}_{t+1}, \boldsymbol{u}^{\prime})$
and $\theta$ is the parameters of the agents. $\bar{\theta}$ is the parameter of the target $Q^{\operatorname{tot}}$ and is periodically copied from $\theta$. $D^{\prime}$ is a sample from the replay buffer $\mathcal{D}$.

\section{\AsyncActionUpperCase Dec-POMDP}\label{sec:async_decpomdp}

We introduce our formulation for \MAASYNCENVS. We first define the \AsyncAction actions\footnote{Asynchronicity is prevalent in real-world multi-agent scenarios, including asynchronicity in observations, actions and communication, etc. In this paper, we focus on the asynchronicity of actions in multi-agent scenarios. For brevity, we name the asynchronicity of actions in MARL as  \textit{off-beat}.} for multi-agent scenarios; then we propose the \AsyncActionUpperCase Dec-POMDP. All the proofs can be found in Appx.~\ref{appendix:proofs}. 

\thmbox{definition}{[\AsyncActionUpperCase Actions] \AsyncActionFirstCap action $\thicktilde{u} \in \mathcal{U}$ characterizes \MAASYNCENVS~where the action $\thicktilde{u}_i$ taken by agent $i$ has execution duration $m_{\thicktilde{u}_i} \sim A(m | \thicktilde{u}_i, i)$ , $A \in \mathcal{A}$, $m \in \{0, 1, 2, \cdots, M\}$ and $M \leq  T$, where  $T$ is the maximum duration and $A$ 
is the action duration distribution. It is a distribution and takes $\thicktilde{u}_i$ and the index of the agent as parameters. $A$ can be either stochastic or deterministic. The joint \AsyncAction action is $\bm{\thicktilde{u}} = [\thicktilde{u}_i]^{N}_{i=1}$. The execution duration is decided at the time the action was committed to the environment. Thus, the execution duration of an action $\bm{\thicktilde{u}}_{t}$ initiated at timestep $t$ is $\bm{m}_t=\{m^{t}_{\thicktilde{u}^{t}_i}\}^{N}_{i=1}$.
}
Note that for each agent, $m^{t}_{\thicktilde{u}^{t}_i}$\footnote{We will omit $t$ in the rest of the paper for brevity.} can be different. At timestep $t$, there are at least $1$ action \footnote{We note that agents have a special NO-OP action available.} and at most $N$ actions being initiated (committed to the environment for execution), leading to asynchronicity of the joint actions.
Next, we propose the \AsyncActionUpperCase Dec-POMDP for \MAASYNCENVS~and discuss its properties. 
\thmbox{definition}{[\AsyncActionUpperCase Dec-POMDP]\label{def:async_decpomdp}\AsyncActionUpperCase Dec-POMDP extends Dec-POMDP, such that \\
(1) state space is $\mathcal{S}$; (2) joint action space is $\mathcal{U}^{N}$; (3) action duration space is $\mathcal{A}^N$; \\
(4) transition function is $\mathcal{P}(\bm{s}'|\bm{s},\thicktilde{\boldsymbol{u}}, \bm{m}):\mathcal{S} \times \mathcal{U}^N \times \mathcal{S} \times \mathcal{A}^N \mapsto [0,1]$, and $\bm{m}$ is the action durations of the joint action; \\
(5) the reward function is $R(\bm{s},\thicktilde{\boldsymbol{u}}, \bm{m}): \mathcal{S} \times \mathcal{U}^N \times \mathcal{A}^N \mapsto \mathcal{R}$; \\
(6) we call a reward $r$ as \AsyncAction reward when any its $m_{\thicktilde{u}_i} \geq 1$, $m_{\thicktilde{u}_i} \in \boldsymbol{m}$, and $r \neq 0$.}

In \MAASYNCENVS, at each timestep $t$, the environment receives actions that agent initiates for execution in the environment. The initiated actions $\thicktilde{\boldsymbol{u}}_t$ are instantaneous actions inferred by agents' policies given individuals' observations. The joint reward is the consequence of the committed joint actions of current timestep and previous timesteps, depending on the actions' duration.
The asynchronicity is an inherent feature of the environment, which is different from asynchronicity incurred by communication delays in many video games (asynchronous gameplay\footnote{\url{https://www.whatgamesare.com/2011/08/synchronous-or-asynchronous-definitions.html}}). We discuss some properties of \AsyncActionUpperCase Dec-POMDP below. 

\begin{remark}
\textit{
When the durations for all actions are identical, \AsyncAction Dec-POMDP reduces to Delayed Dec-POMDP and there is no \AsyncAction actions in it.}
\end{remark}

\begin{remark}
\textit{
There exists $\thicktilde{\boldsymbol{u}}$ that is synchronous since duration of agents' actions can be $m=0$. When $m$ of all actions is $zero$, \AsyncAction Dec-POMDP reduces to Dec-POMDP.} 
\end{remark}
In Delayed Dec-POMDP, actions have the same delayed timesteps, which is different from \AsyncAction actions where actions have different action durations or delays. In order to investigate the problem, we consider the deterministic setting of the transition function and the reward function.


\begin{remark}[Non-episodic Reward]
\textit{
In our formulation, the reward is not episodic reward~\cite{han2021off}.}\label{rmrk:non_episodic_rew}
\end{remark}


\begin{remark}[Non-Markovian Reward]\label{theory:non_markov}\textit{With \AsyncAction actions, the Markovian property of the reward function $R(\bm{s},\thicktilde{\boldsymbol{u}}, \bm{m})$ does not hold.}
\end{remark}

With \AsyncAction actions, the shared rewards can be readily displaced, causing non-Markovian rewards. 
Solving \AsyncActionUpperCase Dec-POMDP is challenging as discussed in Sec. \ref{sec:intro}. 
We propose our methods to tackle aforementioned challenges. 

\section{The Journey is the Reward: A Collective Mental Time Travel Method} \label{sec:our_method}

We propose two methodological elements for \AsyncActionUpperCase MARL.
The first, \OurMethod, presented in this section, is a form of episodic memory that facilitates discovery of a pivotal timestep for \AsyncAction rewards; and the second, presented in Sec.~\ref{sec:redistribution}, is redistribution of the \AsyncAction reward to the pivot timestep when the relevant \AsyncAction actions were initialised.



\subsection{\OurMethod: A \underline{Le}velled \underline{G}raph \underline{E}pisodic \underline{M}emory for \AsyncActionUpperCase MARL}\label{sec:triem}

Human learning relies on retrospecting our detailed memory of the past~\cite{tulving1985memory,suddendorf2009mental}. For example, while exploring a new scenic area, we do not just remember a multitude of specific spots there, but can recall the paths that connect them with junctions and turns. 
However, there is no MARL method that can explicitly recall the past and identify key states that lead to future rewards. Such ``\textit{mental time travel}''~\cite{lampinen2021mentaltime} ability is vital for tackling the challenges in \MAASYNCENVS. Inspired by the recent progress in RL with episodic memory~\cite{hung2019optimizing,botvinick2019reinforcement,HuHaoGEM2021} that is based on the memory prosthesis proposed by neuroscientists~\cite{tulving1985memory,suddendorf2009mental}, 
we propose our method of episodic memory representation for MARL. Unlike previous episodic memory methods that train a parameterized memory by either augmenting the policy inputs for execution~\cite{hung2019optimizing} or regularizing the TD learning~\cite{HuHaoGEM2021} for RL, our method utilizes the levelled graph data structure~\cite{biggs1986graph}, a well established structure for data storage and retrieval, to represent an agent's individual episodic memory. 

\begin{wrapfigure}{r}{0.27\textwidth}
\vspace{-0.8cm}
    \begin{center}
        \includegraphics[width=.27\textwidth]{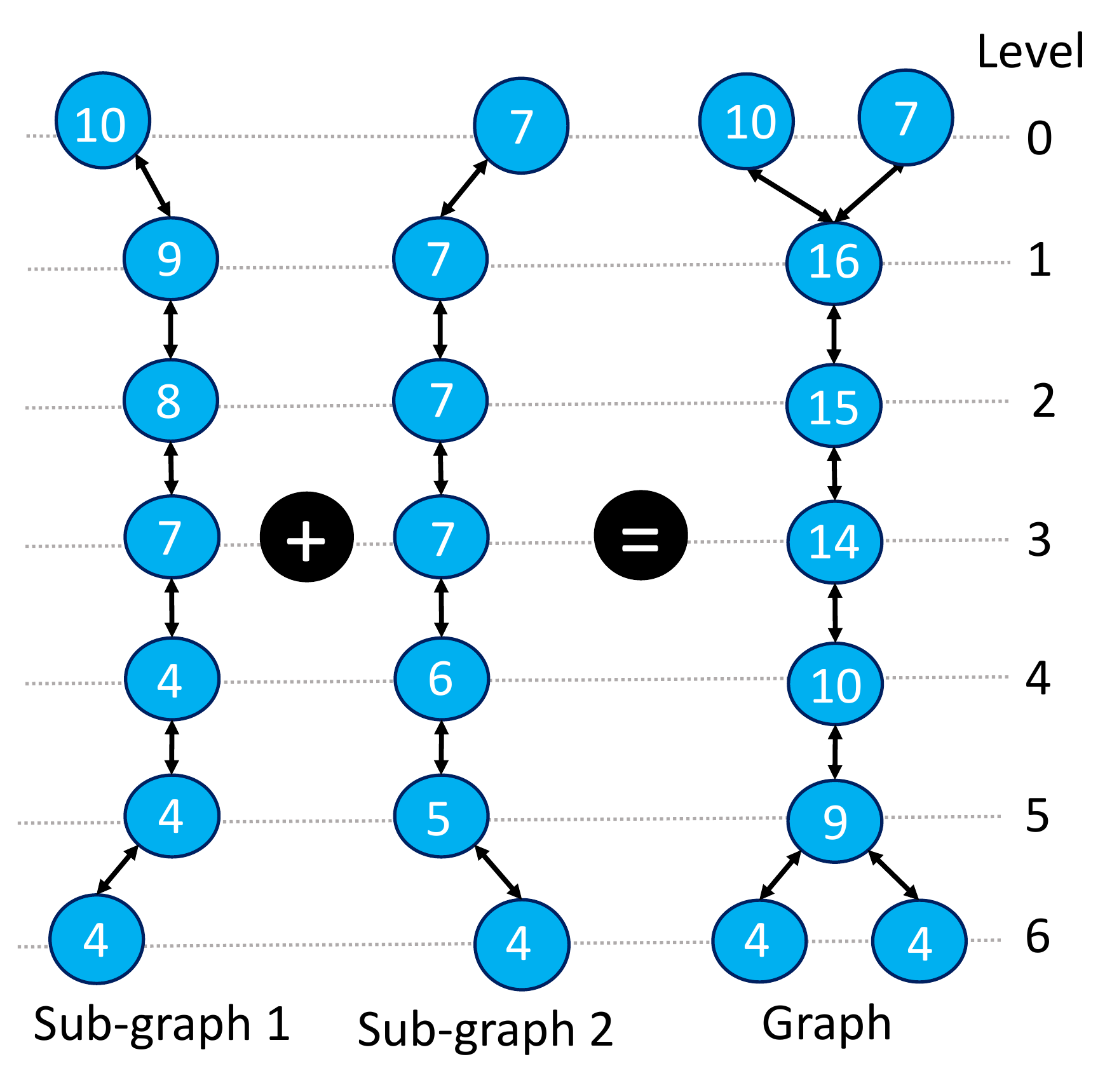}
        \vspace{-0.4cm}
        \caption{\footnotesize{The maximum level of the graph is 7. Circles indicate the nodes and numbers indicate the visit count.}}\label{fig:subgraphs}
    \vspace{-0.5cm}
    \end{center}
\end{wrapfigure}
We propose our novel episodic memory,  Levelled Graph Episodic Memory method (LeGEM), via the levelled graph data structure. \OurMethod~memorizes each agent's past trajectories which are partial observations and the unilateral action of the agent. During training, each agent $i$ collects its individual trajectories $\tau_{i}$. We then define $\tau_{i}$ of agent $i$ as $ \tau_{i}=[(o^{0}_{i},\thicktilde{u}^{0}_{i}, r^{0}), \cdots, (o^{T-1}_{i},\thicktilde{u}^{T-1}_{i}, r^{T-1})]$,
where $T$ is the length of the trajectory and the triplet $(o^{t}_{i},\thicktilde{u}^{t}_{i}, r^{t})$ represents the observation, action and reward of timestep $t$. Note that $r^{t}$ is globally shared between agents. We define agent $i$'s \OurMethod~as a directed graph $\phi^{t}_{i} \in \Phi_{i}$ where $\Phi_i$ is the set of graphs of agent $i$ and $\phi^{t}_{i}$ is the $t$-th graph of $\Phi_{i}$, $t \in \{0, \cdots, T-1\}$. Each $\phi^{t}_{i}$ consists of a tuple of $\left( \Psi, \Xi \right)$ where $\Psi$ is the set of nodes and $\Xi$ represents the set of edges that connect nodes in the graph. 
To model an agent's behaviour explicitly and make the trajectories 
of agents easy to represent,
we create $T$
graphs for each agent and let
$\Phi_i=\{\phi^{t}_i\}^{T-1}_{t=0}$ where $T$ is the maximum level of all graphs and the maximum length of the episode as well. The maximum level of $\phi^{t}_i$ is $t+1$. The node contains key, visit count and pointers connecting the precursors (node at the previous level) and the successors (node at the next level). 
Besides the $\phi^{t}_{i}$, we define the sub-graph set of $\phi^{t}_{i}$ as $\Phi^{t,\Omega}_{i}=\{\phi^{t,\omega}_{i}\}^{\Omega-1}_{\omega=0}$ by using the discretized episode return and there are $\Omega$ sub-graphs. $\phi^{t,\omega}_{i}$ is the $\omega$-th sub-graph whose episode return is $\Upsilon \texttt{[} \omega \texttt{]}$ ($\omega \in \{0, \cdots, \Omega-1\}, \Upsilon=[0, \cdots, \boldsymbol{r}^{t,i}]$) where $\boldsymbol{r}^{t,i}$ is the discretized maximum episode return of $\phi^{t}_{i}$.
Unlike many parameterized episodic memory using state/observation as the key~\cite{hung2019optimizing,lampinen2021mentaltime},
we resort to \textit{afterstate}~\cite{powell2007approximate}. That is, we use agent $i$'s observation $o^{t}_{i}$ and action at timestep $t$, $\thicktilde{u}^{t}_{i}$, to define the key $(o^{t}_{i}, \thicktilde{u}^{t}_{i})$. We provide an example to showcase the relationship between sub-graph and the graph in Fig. \ref{fig:subgraphs}. For complex and continuous state scenarios, for example StarCraft II scenarios, we use SimHash~\cite{charikar2002similarity} to discretize the key $(o^{t}_{i}, \thicktilde{u}^{t}_{i})$.
This technique has been widely used in commercial search engines and RL~\cite{tang2017exploration}. Visit count indicates the total visits made by agent $i$ to the node. It initial value is $1$. Note that nodes are bidirectional since it is helpful for searching (see Sec. \ref{sec:find_the_pivot_time}).

Given a $\tau_i$ with the length of $T$, if the node is already in the graph at level $t$, we then increase the visit count by 1. 
Otherwise, we create a new node for level $t$ of the graph and update its pointers. Meanwhile, sub-graphs will be also created and updated. The process of updating LeGEM is in Alg. \ref{alg:update_graph}. We provide an example of Alg. \ref{alg:update_graph} in Fig. \ref{fig:construct_memory}, Appx.~\ref{appendix:ourmethod}. It is worth noting that $\tau_{i}$ is generated via the interaction of the agent with the environment, and there is no extra interaction needed to collect $\tau_{i}$. The generated trajectories are saved in the experience replay and later sampled for MARL training.

\begin{figure}[t!]
\begin{minipage}[h]{7.2cm}
\begin{algorithm2e}[H]
\textbf{Input:} $\tau$, $\Phi$, $\Upsilon$ and \texttt{Search} (scheme I or II);\\
\textbf{Initialize:} $\kappa$: an empty list to store pivot timesteps; \\
\tcp{Length of $\tau$ and $\tau_i$ are equal.}
$l \leftarrow \texttt{length(}\tau_i\texttt{)-1}$;\\
\For{$t\leftarrow0$ {\bfseries to} $\normalfont{\texttt{length(}}\tau\normalfont{\texttt{)-1}}$} {
    \If{$r^{t} \neq 0~(r^{t} \in \tau)$}{\tcp{\footnotesize{\AsyncActionFirstCap reward}}
        \For{$i\leftarrow1$ {\bfseries to} $N$} {
            Get $\tau_i$ from $\tau$; \\ \label{alg:search_extract}
            $\phi^{l}_i\leftarrow \Phi_i \texttt{[}l\texttt{]}$; \\
            $\psi \leftarrow \phi^{l}_i\texttt{.getNode(}o^{t}_{i}, \thicktilde{u}^{t}_{i}\texttt{)}$; \\
            Find all the paths $\Lambda^{t,l}_{i}$ from node $\psi$ to the node at level $0$;\\
            Get the discretized episode return $\boldsymbol{r}^{l,i}$; \\
            Get the index $\omega$ from $\Upsilon$ with $\boldsymbol{r}^{l,i}$; \\
            $e^{i}_{t} \leftarrow \texttt{Search(}\omega, \Lambda^{t}_{i}, \tau_i, \boldsymbol{r}^{l,i}, \Upsilon, \Phi_i \texttt{)}$; \label{alg:search_pivot_search_scheme}\\
        }
        Get $e_{t}$ (Eqn. \ref{eq:calculate_e_t}) and append $e_t$ to $\kappa$; \label{alg:search_get_e_t}
    }
}
\textbf{Return:} $\kappa$. \label{alg:search_return_kappa}
\caption{\texttt{SearchPivotTimesteps} ($\rho$)}\label{alg:search_triem_algo}
\end{algorithm2e}
\end{minipage}
\hspace{0.25cm}
\begin{minipage}[h]{6cm}
\vspace{0pt}
\begin{algorithm2e}[H]
\textbf{Input:} $\omega$, $\Lambda^{t,l}_{i}$, $\tau_i$, $\boldsymbol{r}^{l,i}$, $\Upsilon$ and $\Phi_i$;\\
\textbf{Initialize:} $\boldsymbol{e}^{i}_t$: a list whose values are all $t$ and its size is the number of paths in $\Lambda^{t,l}_{i}$; \\
$\phi^{l,\omega}_i \leftarrow \Phi^{l,\Omega}_i \texttt{[}\omega\texttt{]}$; \\
$\text{vc} \leftarrow \texttt{VisitCount(}\Lambda^{t,l}_{i}\texttt{)}$ (Alg. \ref{alg:create_mask}); \label{alg:search_I_visitcount} \\
\ForEach{\text{path} $\Lambda^{t,l}_{i}\normalfont{\texttt{[}j\normalfont{\texttt{]}}} \in \Lambda^{t,l}_{i}$}{%
    $e^{i,j, \downarrow}_{t} \leftarrow \texttt{UL(}\Lambda^{t,l}_{i}\texttt{[}j\texttt{]}, \text{vc}, \tau_i\texttt{)}$ (Alg. \ref{alg:ul}); \label{alg:search_I_UL} \\
    $e^{i,j, \uparrow}_{t} \leftarrow \texttt{LU(}\Lambda^{t,l}_{i}\texttt{[}j\texttt{]}, \text{vc}, \tau_i\texttt{)}$ (Alg. \ref{alg:lu}); \label{alg:search_I_LU} \\
    \uIf{$e^{i,j, \downarrow}_{t} \neq -1$}{
         $\boldsymbol{e}^{i}_t\texttt{[}j\texttt{]} \leftarrow e^{i,j, \downarrow}_{t}$;\\
    }
    \uElseIf{$e^{i,j, \uparrow}_{t} \neq -1$}{
        $\boldsymbol{e}^{i}_t\texttt{[}j\texttt{]} \leftarrow e^{i,j, \uparrow}_{t}$;\\
    }
    \Else{
        $\boldsymbol{e}^{i}_t\texttt{[}j\texttt{]} \leftarrow t$;\\
    }
}
$e^{i}_t \leftarrow \texttt{Summarize(} \boldsymbol{e}^{i}_t\texttt{)}$ (Alg. \ref{alg:summarize}) \label{alg:search_I_summarize};  \\
\textbf{Return:} $e^{i}_t$.
\caption{Search Scheme I}\label{alg:search_scheme_I}
\end{algorithm2e}
\end{minipage}
\vspace{-0.7cm}
\end{figure}

\subsection{Multi-Agent Collective Mental Time Travel with \OurMethod}\label{sec:find_the_pivot_time}

With structured agent's past experiences, it can be used to search the pivot timestep when actions that triggered the rewarded state were executed. For example, with \OurMethod, we can find the pivot timestep, $e_t=5$, when agent $1$ shoots the arrow in Fig.~\ref{fig:async_example}. 
\begin{fact}
(Action-Reward Association) \textit{When an \textit{\AsyncActionNotilde} reward $r_t$ exists in the trajectory $\tau_i~(i \in \{1, \cdots, N\})$, $r_t \in \tau_i$, \textit{\AsyncActionNotilde} action $\boldsymbol{u}_{t^{\prime}}$ exists ($t^{\prime} < t$) in the trajectory set $\{\tau_j\}^{N}_{j=1}$, where $\{\tau_j\}^{N}_{j=1}$ constitutes the global trajectory of all agents.\label{fact:action_reward_association}}
\end{fact}
As the reward function and transition function are deterministic in our setting, Fact~\ref{fact:action_reward_association} holds. Intuitively, once we find an \AsyncAction reward in a trajectory, we are sure that the action which triggered the reward can be found in the trajectory. With more experiences collected by the agents, such pattern is obvious and significant. It motivates us to propose a method to leverage the association property of the \AsyncAction action-reward data and search the pivot timestep for timesteps when \AsyncAction rewards occur, which can further help to redistribute the reward backward to mitigate the temporal credit assignment issue (c.f. Sec.~\ref{sec:redistribution}). Therefore, we first propose a search method to search the pivot timestep and then propose a proximal ranking method to estimate the pivot timestep that invokes the future reward. 

\begin{figure}[ht]
    \vspace{-0.5cm}
    \centering
    \includegraphics[scale=0.415]{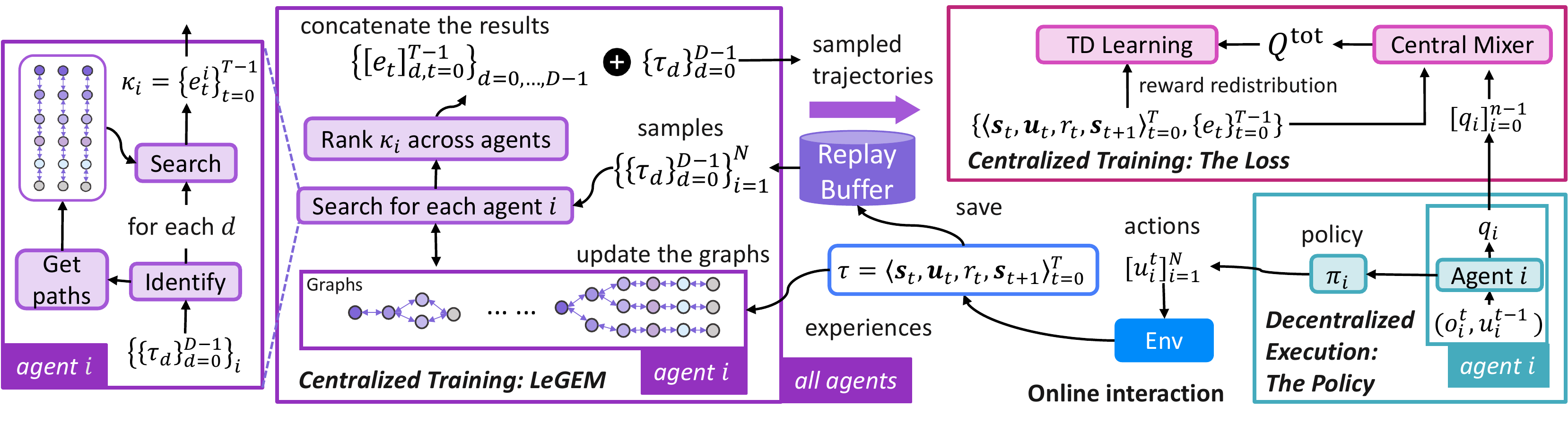}
    \vspace{-0.5cm}
    \caption{\footnotesize{Our framework: \OurMethod, the loss and the agent's policy.}}\label{fig:ctde_framework}
    \vspace{-0.5cm}
\end{figure}
\textbf{Collective Mental Time Travel.} The displaced rewards in the replay buffer hinder multi-agent learning. 
It is essential for each agent to search the pivot timestep when the potential \AsyncAction action that triggered the rewarded state was committed to the environment. Therefore, we propose two search schemes to find the pivot timestep for all agents given an \AsyncAction reward. 
 
\textit{Scheme} I: For agent $i$, given $r_{t} \in \tau_i$, episode return $\boldsymbol{r}^{l,i}$ of $\tau_i$, $\phi^{l}_i = \Phi_i \texttt{[}l\texttt{]}$ and $\phi^{l,\omega}_i = \Phi^{l,\Omega}_i \texttt{[}\omega\texttt{]}$ 
, agent $i$ searches from the node (the key is $(o^{t}_{i}, \thicktilde{u}^{t}_{i})$ and $o^{t}_{i} \in \tau_i$, $u^{t}_{i} \in \tau_i$) at level $t$ in sub-graph $\phi^{l,\omega}_i$ to find the pivot timestep $e_t$ for $r_t$. Concretely, we propose our bi-directional search method. The  first one is called Low-Up (LU) search, which traverses from the given node at level $t$ upwards to the node at level $0$. The second one is named Up-Low (UL) search which
traverses from the node at level $0$ downwards to the given node at level $t$. LU traversing ends when the pattern of increasing visit count
ends and the corresponding level is the candidate pivot timestep. On the contrary, UL traversing ends when the pattern of decreasing visit count ends and the corresponding level is the candidate pivot timestep. In Alg. \ref{alg:search_scheme_I}, we first get visit count (Line \ref{alg:search_I_visitcount}) and then apply UL traversing (Line \ref{alg:search_I_UL}) and LU traversing (Line \ref{alg:search_I_LU}). We summarize the results (Line \ref{alg:search_I_summarize}) by select the pivot timestep that has the maximum count. UL traversing has a higher searching priority than its counterpart. The reason is that there exists pattern that the visit count is decreasing from the node at level $0$ and such pattern ends at the pivot timestep. In practise, it works well in scenarios whose trajectories are single-\AsyncActionNotilde-reward trajectories (there is only one \AsyncAction reward) and the accuracy of Scheme I is over 90\% in grid world scenarios. For scenarios, especially complex scenarios, whose trajectories are multiple-\AsyncActionNotilde-reward trajectories, we apply Scheme II. We put Alg.~\ref{alg:create_mask}, Alg.~\ref{alg:ul}, Alg.~\ref{alg:lu} and Alg.~\ref{alg:summarize} in Appx.~\ref{appendix:ourmethod} as these algorithms are intuitive and easy to understand literally. The time complexity is $\mathcal{O}(n \cdot m)$ (a slight notation abuse) where $n$ is the size of each $\Lambda^{t,l}_{i}$ and $m$ ($1\leq m \leq n$) is the average distance between the level of the given node to the level of the node at the pivot timestep. 

\textit{Scheme} II: Scheme II is a simplified version of scheme I for scenarios that have multiple-\AsyncActionNotilde-reward trajectories, which searches the pivot timestep by finding the nearest timestep in the most visited path. The node of the nearest timestep has the maximum visitcount in that path. Despite the simplicity, it works effective and the time complexity is $\mathcal{O}(n)$ where $n$ is the number of paths in $\Lambda^{t,l}_{i}$. The pseudo code is shown in Alg. \ref{alg:search_scheme_II} in Appx.~\ref{appendix:ourmethod}.

Given a node at level $t$, agents collectively search from the node to find the pivot time step (Line \ref{alg:search_pivot_search_scheme} in Alg. \ref{alg:search_triem_algo}). The visit count is vital for search methods. In MARL, we use $\epsilon$-greedy~\cite{mnih2015human} for agents to explore the environment and collect individual trajectories. The collected trajectories will be used to build the memory and train the policy. We apply annealing to  $\epsilon$ (in Appx.~\ref{appendix:experiments}).

\textbf{Ranking the Pivot Timesteps.} With our two search schemes, we can search the pivot timesteps for each global trajectory $\tau=\{(\boldsymbol{s}^{t}, \thicktilde{\boldsymbol{u}}^{t}, r^t, \boldsymbol{s}^{t+1})\}^{T-1}_{t=0}$. We define the pivot timesteps $\kappa$ of each global trajectory $\tau$ as 
$\kappa=\{e_{t}\}^{T-1}_{t=0}, ~0 \leq e_{t} \leq t$,
where $e_{t}$ indicates the pivot timestep of $t$ when $r_t$ is the consequence of actions committed before timestep $t$. We first get $e_{t}$ by aggregating all the searching outcomes (Line \ref{alg:search_pivot_search_scheme} in Alg. \ref{alg:search_triem_algo}). 
Then, each agent gets $\kappa_i=\{e^{i}_{t}\}^{T-1}_{t=0}$. In order to subserve  the inter-agent credit assignment~\cite{foerster2017counterfactual,rashid2018qmix}, $\kappa$ can be collectively calculated via proximity:
\begin{small}
\begin{equation}\label{eq:calculate_e_t}
    e_{t} = \min_{e^{i}_{t}}\left[ t-e^{1}_{t}, \cdots, t-e^{N}_{t}\right], ~ i \in \{1, \cdots, N\}
\end{equation}
\end{small}\noindent
The pseudo code is shown in Alg. \ref{alg:search_triem_algo}. For each sampled global trajectory $\tau$, we extract $\tau_i$ for each agent in Line \ref{alg:search_extract}; then we get $e_t$ for each agent and aggregate $\kappa$ in line \ref{alg:search_get_e_t} and line \ref{alg:search_return_kappa}, respectively.

\section{Reward Redistribution for \AsyncActionUpperCase Multi-Agent Reinforcement Learning}\label{sec:redistribution}
Searching in \OurMethod~leverages the collective intelligence~\cite{leibo2021meltingpot,ha2021collective} in \MAASYNCENVS. 
We utilize TD learning to train MARL policies. The TD error is the difference between the TD target and the prediction. TD targets can be estimated with $n$-step target, TD($\lambda$) and other techniques~\cite{espeholt2018impala,van2021expected}. Unfortunately, current $n$-step target and TD($\lambda$) methods are far from accurate estimating TD targets. They even incur underestimation with \AsyncAction trajectories. In essence, to train MARL policies in \MAASYNCENVS, one should accurately estimate the TD target where the reward plays the key role~\cite{silver2021reward,zahavy2021reward}. We resolve the aforementioned conundrum by redistributing rewards to their pivot timesteps. The key idea is that we can pull the outcome of one joint \AsyncAction action back to the timestep when it was committed to the environment, which can dramatically enhance learning despite the long-term reward delays incurred by \AsyncAction actions.
We utilize $e_t$ to update the reward  of the transit $(\boldsymbol{s}^{e_t}, \thicktilde{\boldsymbol{u}}^{e_t}, r^{e_t}, \boldsymbol{s}^{e_t+1})$:
\begin{small}
\begin{equation}
    \hat{r}^{e_t}=\mathbbm{1}(e_t \geq t) \cdot r^{e_t}+ \mathbbm{1}(e_t < t) \cdot r_{t}, ~\text{and}~ r_t=(1-\mathbbm{1}(e_t < t) \cdot (1-\beta)) \cdot r_t. \label{eq:rew_redist}
\end{equation}
\end{small}\noindent
where $\mathbbm{1}(\cdot)$ is the indicator function. Such update rule is conducted iteratively from $t=0$ to $t=T-1$. $\beta$ is a very small positive hyperparameter. To stabilize learning and circumvent the overestimation of the TD target and discard the $r^{e_t}$ when $\mathbbm{1}(e_t < t)$ is true, $r_t$ is also updated in Eqn. \ref{eq:rew_redist} to avoid aggregated biased/wrong estimation of TD target being back propagated in Bellman Equation. Formally, we define the reward redistribution operator as $\Pi_{\Phi}$, \ie, $e_t=\Pi_{\Phi} \rho (r^{t}, \boldsymbol{s}, \thicktilde{\boldsymbol{u}})$, and then define the \AsyncActionUpperCase Bellman operator $\Gamma$:
\begin{small}
\begin{equation}
    (\Gamma Q^{\mathrm{tot}})(\boldsymbol{s}, \thicktilde{\boldsymbol{u}}) := \mathbb{E}[ \Pi_{\Phi}R(\boldsymbol{s}, \thicktilde{\boldsymbol{u}}, \boldsymbol{m})  + \gamma \max_{\thicktilde{\boldsymbol{u}}^{\prime}} Q^{\mathrm{tot}}(\boldsymbol{s}^{\prime}, \thicktilde{\boldsymbol{u}}^{\prime}) ]
\end{equation}
\end{small}\noindent
With the \AsyncActionUpperCase Bellman operator $\Gamma$, we propose its contraction property. 
\begin{restatable}{proposition}{PropContract}\label{propContract}
$\Gamma: \mathcal{Q} \mapsto \mathcal{Q}$ is a $\gamma$-contraction.
\end{restatable}
Therefore, we can utilize $\hat{r}_{e_t}$for \textit{centralized training} in TD-learning:
\begin{small}
\begin{equation}
    \mathcal{L}^{\operatorname{TD}} (\theta) :=\mathbb{E}_{\mathcal{D}^{\prime} \sim \mathcal{D}}[(\hat{y}_{e_t}^{\text {tot}}-Q_{\theta}^{\text {tot }}(\boldsymbol{s}^{e_t}, \thicktilde{\boldsymbol{u}}^{e_t}))^{2} ],~~\text{where} ~~\hat{y}_{e_t}^{\mathrm{tot}}=\hat{r}^{e_t}+\gamma \max _{\thicktilde{\boldsymbol{u}}^{\prime}} Q_{\bar{\theta}}^{\mathrm{tot}}(\boldsymbol{s}^{e_t+1}, \thicktilde{\boldsymbol{u}}^{\prime}).
\label{eq:new_td_loss_dqn}
\end{equation}
\end{small}\noindent
Our method can be easily incorporated into any model-free MARL method for \MAASYNCENVS. 
We present the pseudo code of incorporating our method into model-free MARL methods in Alg. \ref{alg:triem_algo}, Appx.~\ref{appendix:experiments}.
We also provide a pictorial view of our framework in Fig. \ref{fig:ctde_framework} to show the whole pipeline. 

\section{Experiments}\label{sec:experiments}


We perform experiments on various multi-agent scenarios with \AsyncAction actions. We introduce \AsyncAction actions in Stag-Hunter Game, Quarry Game,  Afforestation Game and StarCraft II micromanagement tasks~\cite{samvelyan19smac} and use them as testbeds in our experiments. We aim to answer the following questions: \textbf{Q1:} \textit{Can our \OurMethod~improve the multi-agent coordination of many MARL methods in \MAASYNCENVS?} \textbf{Q2:} \textit{Can our \OurMethod~outperform previous parameterized episodic memory (EM) for MARL?}
\textbf{Q3:} \textit{Can bootstrapping method of RL help?} \textbf{Q4:} \textit{Can our \OurMethod~outperform the multi-agent exploration and multi-agent risk-sensitive (Ex-Risk) methods?}

\begin{figure}[ht]
    \vspace{-0.3cm}
    \centering
    \includegraphics[scale=0.35]{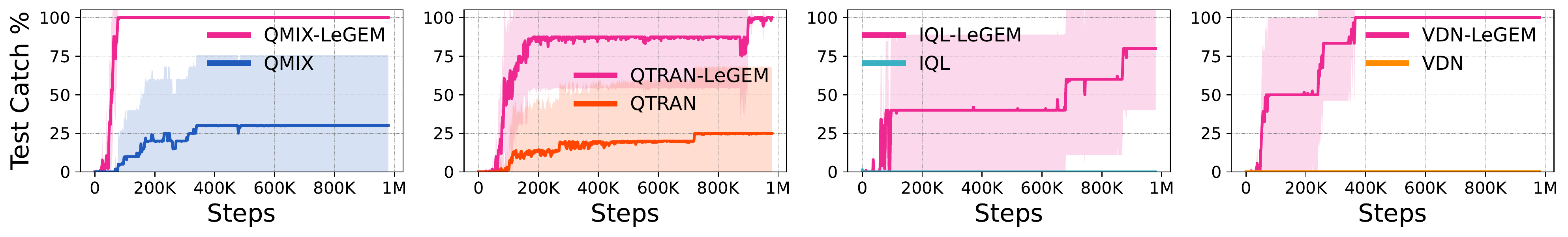}
    \vspace{-0.5cm}
    \caption{\footnotesize{The test catch rate of the stag on the Stag-Hunter Game with \AsyncAction actions.}}
    \label{fig:two_agent_stag_hunter_game}
    \vspace{-0.4cm}
\end{figure}

\subsection{Experiment Setup}

\begin{minipage}{\textwidth}
    \vspace{-0.3cm}
    \begin{minipage}[b]{0.53\textwidth}
    \centering
    \includegraphics[height=3.3cm,width=0.9\textwidth]{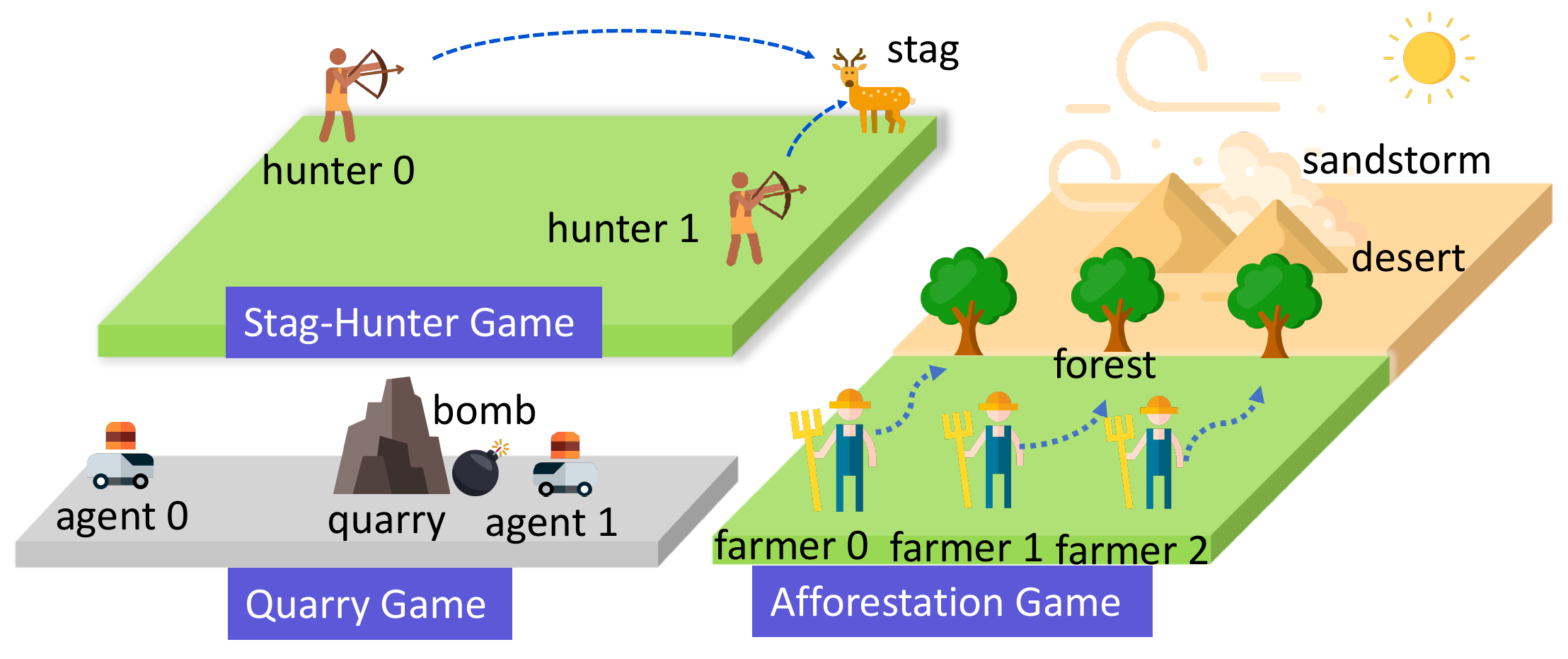}
    \vspace{-0.3cm}
    \captionof{figure}{\footnotesize{Stag-Hunter Game, Quarry Game and Afforestation Game. More information can be found in Appx.~\ref{appendix:environments}.}}
    \label{fig:3_games}
  \end{minipage}
  \hfill
  \begin{minipage}[b]{0.47\textwidth}
    \centering
    \setlength\tabcolsep{2.8pt}
    \renewcommand{\arraystretch}{0.75}
    \begin{tabular}{crcrcr}
        \toprule
        \multicolumn{2}{c}{\footnotesize{Categories}} &
        \multicolumn{2}{l}{\footnotesize{Methods}} \\
        \cmidrule(lr){1-2}
        \cmidrule(lr){3-4}
        \multicolumn{2}{c}{\multirow{3}{*}{\makecell{\footnotesize{MARL (\textbf{Q1})}}}} &  
        \multicolumn{2}{l}{\footnotesize{QMIX}~\cite{rashid2018qmix}, \footnotesize{VDN}~\cite{sunehag2017value}} \\
        \multicolumn{2}{c}{} & \multicolumn{2}{l}{\footnotesize{IQL}~\cite{tampuu2017multiagent}, \footnotesize{QTRAN}~\cite{son2019qtran}} \\
        \multicolumn{2}{c}{} & \multicolumn{2}{l}{\footnotesize{QPLEX}~\cite{wang2020qplex}} \\
        \cmidrule(lr){1-2}
        \cmidrule(lr){3-4}
        \multicolumn{2}{c}{\multirow{1}{*}{\makecell{\footnotesize{EM (\textbf{Q2})}}}} & 
        \multicolumn{2}{l}{\footnotesize{EMC}~\cite{zheng2021EMC}} \\
        


        \cmidrule(lr){1-2}
        \cmidrule(lr){3-4}

        \multicolumn{1}{c}{\multirow{1}{*}{\makecell{\footnotesize{Bootstrap (\textbf{Q3})}}}} & 
        \multicolumn{2}{l}{\footnotesize{N-step \&$\lambda$-Return}~\cite{sutton2018reinforcement}} \\

        \cmidrule(lr){1-2}
        \cmidrule(lr){3-4}

        \multicolumn{2}{c}{\multirow{2}{*}{\makecell{\footnotesize{Ex-Risk  (\textbf{Q4})}}}} & 
        \multicolumn{2}{l}{\footnotesize{MAVEN}~\cite{mahajan2019maven}, \footnotesize{EMC}~\cite{zheng2021EMC}} \\
        \multicolumn{2}{c}{} & \multicolumn{2}{l}{\footnotesize{RMIX}~\cite{qiu2021rmix}} \\
        \toprule
    \vspace{-0.5cm}
    \end{tabular}
    \captionof{table}{\footnotesize{Baseline algorithms}.}\label{tab:baselines}
  \end{minipage}
\end{minipage}

\textbf{Baselines and scenarios.} We list all baselines in table \ref{tab:baselines}, including the corresponding research questions to be answered. We implement our method on PyMARL~\cite{samvelyan19smac} and use 10 random seeds to train each method on all environments. We do not use macro-action methods~\cite{xiao2020macro,xiao2020multi} as the baseline because it is hard to make a fair comparison between macro-actions methods and our method. As discussed in Sec. \ref{sec:intro}, macro-actions rely on manually designed macro-actions, \ie, designing the macro-actions by utilizing the simulator settings and domain knowledge, which is different from learning options~\cite{sutton1999between,bacon2017option}. Designing macro-actions is not feasible in scenarios where domain knowledge and simulator settings are unknown, such as the OBMAS scenarios. In OBMAS, the agent has no idea of the durations of other agents’  actions, which is challenging for designing macro-actions. We conduct experiments on Stag-Hunter Game, Quarry Game, Afforestation Game (Fig. \ref{fig:3_games}) and StarCraft II micromanagement tasks~\cite{samvelyan19smac} where \AsyncAction action are introduced. 

\textbf{Training settings.} We use opensourced code of baselines publicly by the corresponding authors on Github in all experiments. We resort to mean-std values as our performance evaluation measurement in all figures where the bold line and the shaded area indicate the mean value and one standard deviation of the episode return, respectively. Readers can refer to Appx.~\ref{appendix:environments}, \ref{appendix:baselines},  \ref{appendix:experiments} and \ref{appendix:exp_results} for more information on our environment, baselines, training method, training platform and empirical results.

\begin{figure}[ht]
    \vspace{-0.4cm}
    \centering
    \includegraphics[scale=0.35]{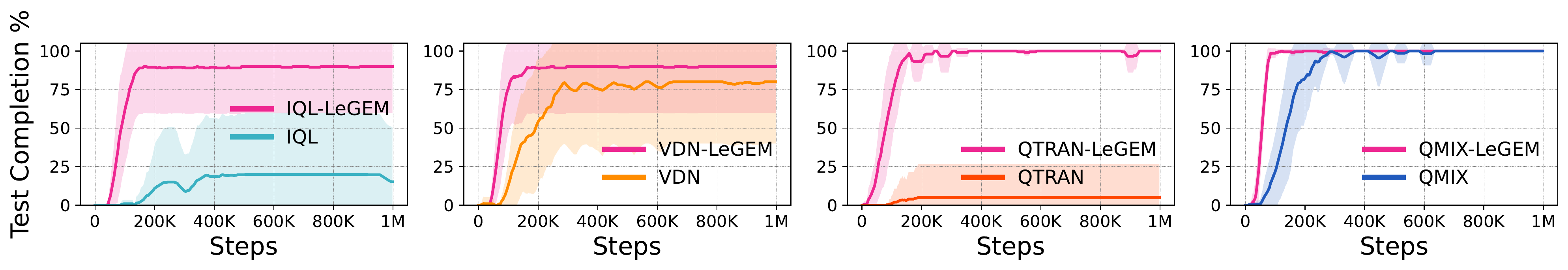}
    \vspace{-0.6cm}
    \caption{\footnotesize{The test task completion rate of the Quarry Game with \AsyncAction actions.}}
    \label{fig:two_agent_quarry_game}
    \vspace{-0.3cm}
\end{figure}
\begin{figure}[ht]
    \vspace{-0.3cm}
    \centering
    \includegraphics[scale=0.375]{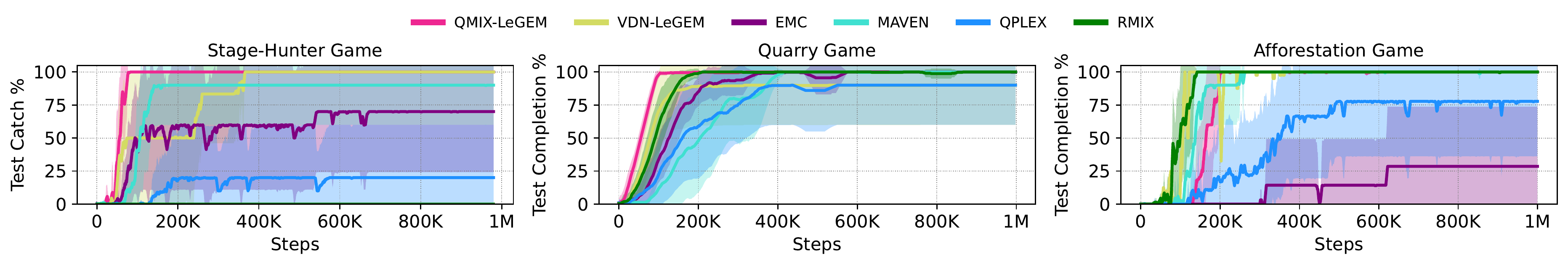}
    \vspace{-0.7cm}
    \caption{\footnotesize{Performance of MARL methods}}\label{baselines_1}
    \vspace{-0.5cm}
\end{figure}
\subsection{Experiment Results}\label{sec:exp_results}

\textbf{The Effectiveness of \OurMethod.} We answer \textbf{Q1}. With \OurMethod, MARL methods get enhanced performance as shown in Fig. \ref{fig:two_agent_stag_hunter_game}. Without \OurMethod, all methods perform poorly in Stag-Hunter Game; IQL and VDN's final final results are even $0$. By incorporating \OurMethod, all of them can get converged performance and improved sample efficiency. We are also interested in finding if \OurMethod~could reinforce the performance of simple methods. As depicted in Fig. \ref{baselines_1}, with \OurMethod, both VDN and QMIX outperforms QPLEX, which is a state-of-the-art MARL method armed with various advanced techniques, including attention network~\cite{vaswani2017attention}, dueling network~\cite{wang2016dueling} and advantage function. 




\textbf{Performance of Episodic Memory method.} We answer \textbf{Q2} by presenting the performance curves of EMC in Fig. \ref{baselines_1}. EMC is an episodic memory MARL method with curiosity-driven exploration. It utilizes the episodic memory from RL~\cite{Zhu2020Episodic,HuHaoGEM2021}.With \OurMethod, QMIX outperforms EMC. EMC even fails to converge in Stag-Hunter Game.


\begin{table}[h]
    \vspace{-0.5cm}
    \setlength\tabcolsep{2.8pt}
    \centering
    \caption{\footnotesize{Results (mean and std) of $n$-step return (left) and TD($\lambda$) (right)  on Stag-Hunter Game.}}\label{n_step_lambda_stag_hunter}
    \begin{tabular}{{l}{l}{l}{l}{l}|{l}{l}{l}{l}{l}}
    \hline
     $n$ & $1$ & $5$ & $10$ & $15$ & $\lambda$ & $0.8$ & $0.9$ & $0.99$ & $1$\\
    \hline\hline
    \footnotesize{QMIX} & \footnotesize{$60.0\pm40\%$} & \footnotesize{$0\pm0$} & \footnotesize{$0\pm0$} & \footnotesize{$0\pm0$} & \footnotesize{QMIX} & \footnotesize{$100\pm0\%$} & \footnotesize{$100\pm0\%$} & \footnotesize{$89\pm10\%$} & \footnotesize{$61\pm37\%$} \\
    \hline
    \footnotesize{VDN} & \footnotesize{$0\pm0$} & \footnotesize{$0\pm0$} & \footnotesize{$0\pm0$} & \footnotesize{$0\pm0$} & \footnotesize{VDN} & \footnotesize{$0\pm0$} & \footnotesize{$0\pm0$} & \footnotesize{$0\pm0$} & \footnotesize{$0\pm0$}\\
    \hline
    \end{tabular}
    \vspace{-0.3cm}
\end{table}

\textbf{Performance of $n$-step return and TD($\lambda$) methods.} To answer \textbf{Q3}, we use $n$-step return and TD($\lambda$) to estimate the TD-target. As shown in Table. \ref{n_step_lambda_stag_hunter}, with $n$-step return, both QMIX and VDN fail to learn good policies even with $n=15$. Surprisingly, with TD($\lambda$), QMIX can achieve good performance with $\lambda \in \{0.8, 0.9, 0.99, 1\}$. However, we cannot find such outcome on VDN and there is no guarantee of good results on using TD($\lambda$). 

\textbf{Performance of Multi-Agent Exploration and Risk-Sensitive MARL methods.} We also provide results of exploration methods for MARL and risk-sensitive MARL method to answer \textbf{Q4}. MAVEN utilizes mutual information to learn latent space for exploration and RMIX aims to learning risk-sensitive policies for MARL.
In Fig. \ref{baselines_1}, RMIX even fails to learn.
Mainly because the potential loss of reward is displaced by \AsyncAction actions. 
Overall, MAVEN is 
stabler than EMC and RMIX.
QMIX-\OurMethod~is stable in all scenarios and outperforms MAVEN. 
With \OurMethod, even simple method such VDN can
perform well and outperforms many MARL methods with complex and advanced components. Indeed, exploration in \MAASYNCENVS~is beneficial for multi-agent learning. However, the key challenge of temporal credit assignment can not be easily addressed merely with exploration.

\begin{wrapfigure}{r}{0.39\textwidth}
    \vspace{-0.5cm}
    \centering
    \includegraphics[scale=0.5]{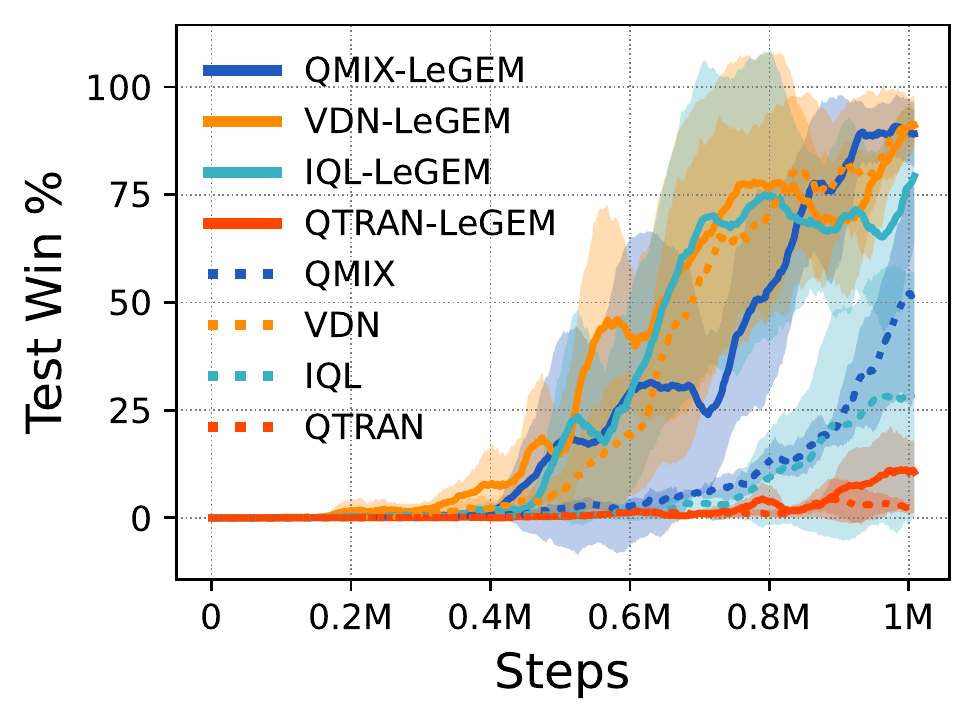}
    \vspace{-0.5cm}
    \caption{\footnotesize{The performance of MARL methods on 2m\_vs\_1z.}}
    \label{fig:smac_res_1}
    \vspace{-0.4cm}
\end{wrapfigure}
\textbf{SMAC.} We also conduct experiments on SMAC~\cite{samvelyan19smac}. We train MARL methods and our method on 2m\_vs\_1z where are two agents combating with one opponent. To overcome the issue of high dimension continuous state space, We utilize simhash~\cite{charikar2002similarity} to calculate the hash value of the key. We only select the attack action and set the action duration with 9. As illustrated in Fig. \ref{fig:smac_res_1}, incorporated with our novel episodic memory, QMIX, IQL and VDN illustrate enhanced performance, demonstrating the superiority of our method on complex multi-agent scenarios.

\section{Related Works}\label{sec:related_works}
\textbf{Action Delay in RL.} Conventionally, the execution of actions in RL is instantaneous and the execution duration is neglected. Katsikopoulos et al.~\cite{katsikopoulos2003markov} propose the Delayed MDP where actions have delays and Walsh et al~\cite{walsh2009learning} propose a model-based method for the Delayed MDP.
To optimize the delayed MDP, many RL approaches~\cite{walsh2009learning,ramstedt2019real,xiao2019thinking,yuan2022asynchronous} augment the state space with the queuing actions to be executed into the environment. However, this state-augmentation trick is intractable~\cite{derman2020acting}. Chen et al.~\cite{chen2020delay} extend the delayed MDP~\cite{ramstedt2019real} and propose a Delayed Markov Game. However, the state-augmentation treatment is confined to short delays and neglects the \AsyncAction actions in multi-agent scenarios. Recently, Bouteiller et al.~\cite{bouteiller2020reinforcement} apply replay buffer correction method.
However, the delayed timestep is privileged information. It is not available for agents in many scenarios. Simply applying this single-agent trajectory correction in MARL cannot attain satisfactory performance due to \AsyncAction actions; devising inter-agent trajectory correction methods for \MAASYNCENVS~is non-trivial. 

\textbf{Credit Assignment in RL.} Credit assignment~\cite{sutton1984temporal,sutton1999between} tackles long-horizon sequential decision-making problem by distributing the contribution of each single step over the temporal interval. TD learning~\cite{sutton2018reinforcement} is the most established credit assignment method, which is the basis of many RL methods. RUDDER~\cite{rudder2019} redistributes the episodic return to key timesteps in the episode~\cite{gangwani2020learning,ren2021learning,raposo2021synthetic}. Klissarov et al.~\cite{klissarov2020reward} propose a reward propogation method via graph convolutional neural network~\cite{kipf2016semi}. Another line of works utilize episodic memory (EM)~\cite{pritzel2017neural,botvinick2019reinforcement,zhou2019memory,ma2021state,Zhu2020Episodic} to recall key events and aggregate information of the past for decision-making or learning. 
However, simply applying EM of RL to MARL cannot perform well in \MAASYNCENVS~due to the non-stationarity and the displaced rewards. 

\textbf{Multi-Agent RL.} 
Many MARL methods focus on factorizing the global Q value to train agents' policies via CTDE~\cite{foerster2017counterfactual,sunehag2017value,rashid2018qmix,son2019qtran,wang2020qplex,wang2021dop,pan2021regularized}. 
However, these existing works assume actions are executed synchronously. 
Messias et al.~\cite{messias2013multiagent} propose an event-driven, asynchronous formulation of the multi-agent POMDP. However, the assumption of free communication~\cite{wang2020learning} is limited and the asynchronous execution~\cite{omidshafiei2015decentralized} in the paper is confined to the design of events and did not propose methods on solving challenging credit assignment issue in \MAASYNCENVS. Recently, Amato et al.~\cite{amato2019modeling} and Xiao et al.~\cite{xiao2020macro,xiao2020multi} propose macro-action methods, which are similar to hierarchical methods. Macro-actions are manually designed via abstracting primitive actions. However, macro-action methods mainly focus on macro-action selection during multi-timestep decision-making and assume the environment can use manually pre-defined methods for state transition. Unfortunately, the above works either focus on synchronous actions or defining specific asynchronous execution components with human knowledge. Learning coordination in \MAASYNCENVS~remains a challenge. 

\section{Conclusion}\label{sec:conclusion}

In this paper, we investigate model-free MARL with \AsyncAction actions. 
To address challenges in \MAASYNCENVS, we first propose \AsyncActionUpperCase Dec-POMDP. Then, we propose a new class of episodic memory, \OurMethod, for model-free MARL algorithms. \OurMethod~addresses the challenging temporal credit assignment problem raised by \AsyncAction actions in TD-learning via the novel reward redistribution scheme. We evaluate our method on various \MAASYNCENVS~scenarios. Empirical results show that our method significantly boosts the multi-agent coordination and achieves leading performance as well as improved sample efficiency.

\textbf{Limitations and Future Work.} Searching from a graph-structured episodic memory takes much overhead in \OurMethod. Scaling up \OurMethodTilde to complex \MAASYNCENVS~is our future direction. Recently, there is a growing interest in model-based planing~\cite{schrittwieser2020mastering}. Leveraging \OurMethod~for model-based planning is also our future work. Our paper focuses on Dec-POMDP-based MARL methods. We leave it to future work for investigating \AsyncAction actions in frameworks like Markov Game~\cite{littman1994markov} and MMDP~\cite{boutilier1996planning}. 
We are also interested in finding the merit of our method in real-world problem in our future work, such as scheduling~\cite{mao2019learning} with \AsyncAction settings.

\bibliography{main}
\bibliographystyle{abbrv}

\newpage

\appendix
\newpage
\section{Proofs}\label{appendix:proofs}

\PropContract*
\begin{proof}
Recall that the \AsyncActionUpperCase Bellman operator $\Gamma$ is defined as:
\begin{small}
\begin{equation}
    (\Gamma Q^{\mathrm{tot}})(\boldsymbol{s}, \thicktilde{\boldsymbol{u}}) := \mathbb{E}[ \Pi_{\Phi}R(\boldsymbol{s}, \thicktilde{\boldsymbol{u}}, \boldsymbol{m})  + \gamma \max_{\thicktilde{\boldsymbol{u}}^{\prime}} Q^{\mathrm{tot}}(\boldsymbol{s}^{\prime}, \thicktilde{\boldsymbol{u}}^{\prime}) ]
\end{equation}
\end{small}\noindent
The sup-norm is defined as $\left\lVert Q \right\lVert_{\infty}=\sup_{\boldsymbol{s} \in \mathcal{S}, \thicktilde{\boldsymbol{u}} \in \mathcal{U}} \left\vert Q(\boldsymbol{s}, \thicktilde{\boldsymbol{u}}) \right\vert$. We consider the sup-norm contraction:
\begin{small}
\begin{equation}
    \left\lVert (\Gamma Q^{\mathrm{tot}}_{(1)})(\boldsymbol{s}, \thicktilde{\boldsymbol{u}}) - (\Gamma Q^{\mathrm{tot}}_{(2)})(\boldsymbol{s}, \thicktilde{\boldsymbol{u}}) \right\lVert_{\infty} \leq \gamma \left\lVert Q^{\mathrm{tot}}_{(1)}(\boldsymbol{s}, \thicktilde{\boldsymbol{u}}) - Q^{\mathrm{tot}}_{(2)}(\boldsymbol{s}, \thicktilde{\boldsymbol{u}}) \right\lVert_{\infty}
\end{equation}
\end{small}
We prove:
\begin{equation}
\begin{aligned}
    \left\lVert (\Gamma Q^{\mathrm{tot}}_{(1)})(\boldsymbol{s}, \thicktilde{\boldsymbol{u}}) - (\Gamma Q^{\mathrm{tot}}_{(2)})(\boldsymbol{s}, \thicktilde{\boldsymbol{u}}) \right\lVert_{\infty} 
    &= \max_{\boldsymbol{s}, \thicktilde{\boldsymbol{u}}}\left\vert\gamma \sum\nolimits_{\boldsymbol{s}^{\prime}}\mathcal{P}(\boldsymbol{s}^{\prime}\vert\boldsymbol{s},\thicktilde{\boldsymbol{u}})(\max_{\thicktilde{\boldsymbol{u}}^{\prime}} Q_{(1)}^{\mathrm{tot}}(\boldsymbol{s}^{\prime}, \thicktilde{\boldsymbol{u}}^{\prime}) - \max_{\thicktilde{\boldsymbol{u}}^{\prime}}Q_{(2)}^{\mathrm{tot}}(\boldsymbol{s}^{\prime}, \thicktilde{\boldsymbol{u}}^{\prime})) 
    \right\vert \\
    &\leq \max_{\boldsymbol{s}, \thicktilde{\boldsymbol{u}}} \gamma  \sum\nolimits_{\boldsymbol{s}^{\prime}}\mathcal{P}(\boldsymbol{s}^{\prime}\vert\boldsymbol{s},\thicktilde{\boldsymbol{u}}) \left\vert\max_{\thicktilde{\boldsymbol{u}}^{\prime}}(Q_{(1)}^{\mathrm{tot}}(\boldsymbol{s}^{\prime}, \thicktilde{\boldsymbol{u}}^{\prime}) -Q_{(2)}^{\mathrm{tot}}(\boldsymbol{s}^{\prime}, \thicktilde{\boldsymbol{u}}^{\prime}))  \right\vert \\
    &\leq \max_{\boldsymbol{s}, \thicktilde{\boldsymbol{u}}} \gamma  \sum\nolimits_{\boldsymbol{s}^{\prime}}\mathcal{P}(\boldsymbol{s}^{\prime}\vert\boldsymbol{s},\thicktilde{\boldsymbol{u}}) \max_{\boldsymbol{s}^{\prime\prime}, \thicktilde{\boldsymbol{u}}^{\prime}}\left\vert Q_{(1)}^{\mathrm{tot}}(\boldsymbol{s}^{\prime\prime}, \thicktilde{\boldsymbol{u}}^{\prime}) -Q_{(2)}^{\mathrm{tot}}(\boldsymbol{s}^{\prime\prime}, \thicktilde{\boldsymbol{u}}^{\prime}) \right\vert \\
    &= \max_{\boldsymbol{s}, \thicktilde{\boldsymbol{u}}} \gamma  \sum\nolimits_{\boldsymbol{s}^{\prime}}\mathcal{P}(\boldsymbol{s}^{\prime}\vert\boldsymbol{s},\thicktilde{\boldsymbol{u}}) \left\lVert Q_{(1)}^{\mathrm{tot}} -Q_{(2)}^{\mathrm{tot}} \right\lVert_{\infty}  \\
    &= \gamma \left\lVert Q^{\mathrm{tot}}_{(1)} - Q^{\mathrm{tot}}_{(2)} \right\lVert_{\infty}
\end{aligned}
\end{equation}

\end{proof}

Readers may find that with the reward redistribution operator $\Pi_{\Phi}$, the reward is ordered. Consequently, \AsyncActionUpperCase Bellman equation is reduced to Bellman equation\footnote{When all \AsyncAction rewards are redistributed to the ground true pivot timestep, we can claim this finding.}.

\newpage
\section{\OurMethod~for \AsyncActionUpperCase MARL}\label{appendix:ourmethod_and_marl}

In this section, we list Alg.~\ref{alg:create_mask}, Alg.~\ref{alg:ul}, Alg.~\ref{alg:lu} and Alg.~\ref{alg:summarize} in Sec. \ref{appendix:ourmethod} and present the training pipline for \AsyncActionUpperCase MARL in Alg.~\ref{alg:triem_algo} in Sec.~\ref{appendix_sec:marl_training}. We also present lists of symbols for Dec-POMDP, \AsyncActionUpperCase Dec-POMDP, MARL, \AsyncActionUpperCase MARL and \OurMethod in Tab.~\ref{tb:appendix_notation_decpomdp},~\ref{tb:appendix_notation_marl} and~\ref{tb:appendix_notation_ourmethod}. To make pesude code easy to read, we use Python-like\footnote{\url{https://www.python.org/}} syntax to represent vectors and hashmaps (look-up tables).

\subsection{\OurMethod}\label{appendix:ourmethod}

We also define the sub-graph set of $\phi^{t}_{i}$ as $\Phi^{t,\Omega}_{i}=\{\phi^{t,\omega}_{i}\}^{\Omega-1}_{\omega=0}$ by using the discretized episode return and there are $\Omega$ sub-graphs. $\phi^{t,\omega}_{i}$ is the $\omega$-th sub-graph whose episode return is $\Upsilon \texttt{[} \omega \texttt{]}$ ($\omega \in \{0, \cdots, \Omega-1\}, \Upsilon=[0, \cdots, \boldsymbol{r}^{t,i}]$) where $\boldsymbol{r}^{t,i}$ is the discretized maximum episode return of $\phi^{t}_{i}$.

\begin{algorithm2e}[H]
\textbf{Input:} Agent $i$'s  $\{\tau^{d}_i\}^{D}_{d=1}$ and $\Phi_i$.\\
\For{$d\leftarrow1$ {\bfseries to} $D$} { 
    Get $\phi^{l}_i \leftarrow \Phi_i \texttt{[length(}\tau^{d}_i\texttt{)-1]}$; \tcp{$\texttt{length(}\tau^{d}_i\texttt{)-1}$ equals $l$}
    Calculate the discretized episode reward $\boldsymbol{r}^{l,i}$;\\
    Get the index $\omega$ from $\Upsilon$ by using $\boldsymbol{r}^{l,i}$; \\
    Get $\phi^{l,\omega}_i \leftarrow \Phi^{l,\Omega}_i \texttt{[}\omega\texttt{]}$; \\
    \For{$t\leftarrow0$ {\bfseries to} $\operatorname{length}(\tau^{d}_i)-1$} {
        \eIf{$(o^{t}_{i}$, $u^{t}_{i}) \in \phi^{l}_i$} {
        
            \tcp{There is no need to update the node of sub-graph $\phi^{l,\omega}_i$ as it shares the same node with $\phi^{l}_i$}
        
            $\psi \leftarrow \phi^{l}_i\texttt{.getNode(}o^{t}_{i}$, $u^{t}_{i}\texttt{)}$; \\
            $\psi\texttt{.visitCount++}$; \\
        } {
            $\psi\leftarrow \texttt{newNode(}o^{t}_{i}, u^{t}_{i}, r^{t}\texttt{)}$; \\
            $\phi^{l}_i\texttt{.append(}\psi\texttt{)}$; \\
            $\phi^{l}_i\texttt{.updatePointers(}\psi\texttt{)}$; \tcp{Sub-graph $\phi^{l,\omega}_i$ shares the same node with $\phi^{l}_i$}
            $\phi^{l,\omega}_i\texttt{.append(}\psi\texttt{)}$; \\ 
            $\phi^{l,\omega}_i\texttt{.updatePointers(}\psi\texttt{)}$;
            
        }
    }
}
\textbf{Return:} $\Phi_i$.
\caption{\texttt{Update\OurMethod}}\label{alg:update_graph}
\end{algorithm2e}

Alg.~\ref{alg:update_graph} shows the whole procedure to construct the graph. To illustrate it, we provide an example below (Fig.~\ref{fig:construct_memory}) to show how to construct the graph. Fig.~\ref{fig:subgraphs} shows the relationship between sub-graphs and the graphs.

\begin{figure}[ht]
    \centering
    \includegraphics[scale=0.45]{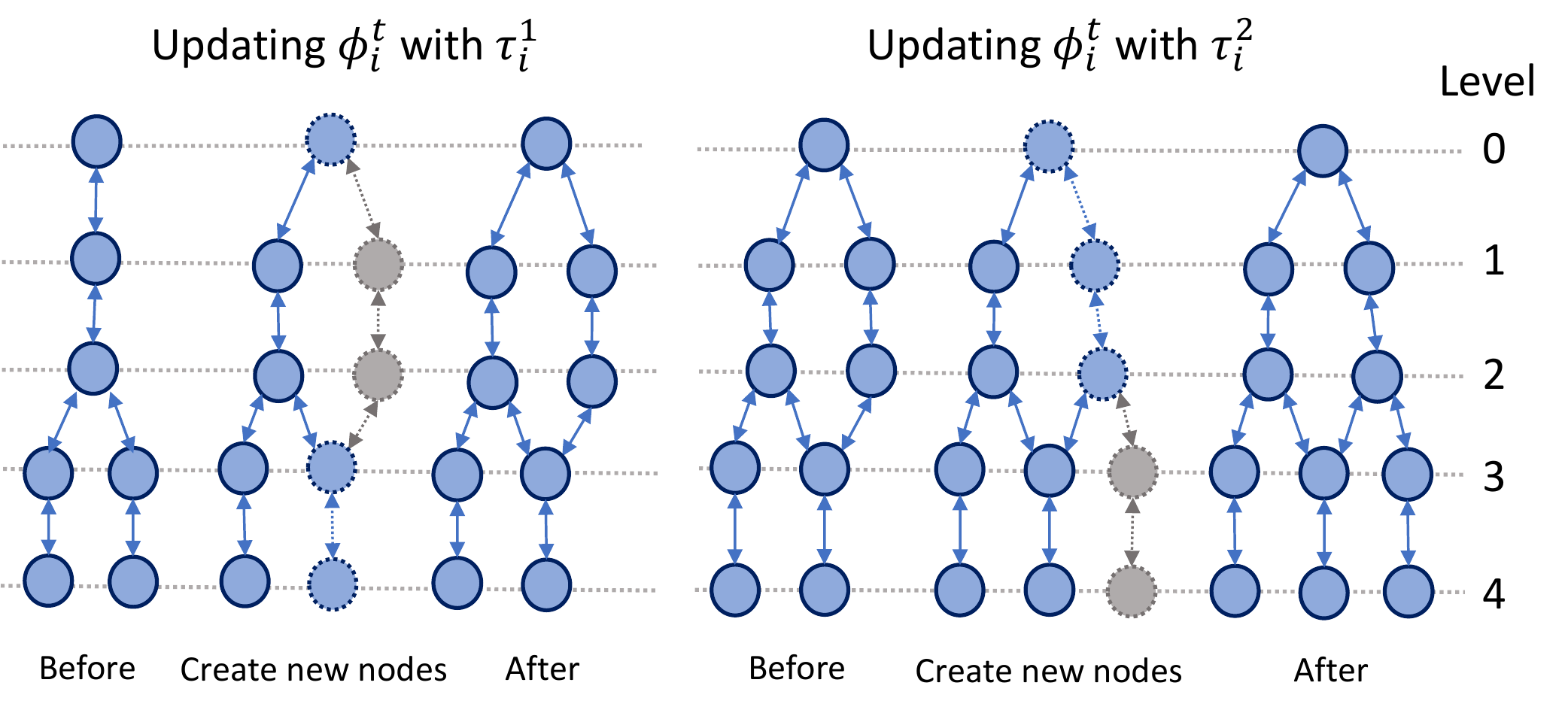}
    \caption{\textbf{Updating agent $i$'s $\Phi_i$:} Agent $i$'s $\phi^{t}_{i}$ is updated with agent's trajectories $\tau^{1}_{i}$ and then updated with $\tau^{2}_{i}$. Solid arrows and circles indicate the pointers and nodes, respectively. Grey dotted lines indicate pointers to be created and grey circles with dotted outlines indicate nodes to be created. All the dotted elements (pointers and circles) consist of the a new path in $\tau^{1}_{i}$. All the created pointers and nodes will be added to $\phi^{t}_{i}$.}
    \label{fig:construct_memory}
\end{figure}

\newpage

\begin{algorithm2e}[H]
\textbf{Input:} $\Lambda^{t,l}_{i}$.\\
\textbf{Initialize:} $\text{vc} \leftarrow \texttt{[]}$, an empty vector to store visit count mask for each path in $\Lambda^{t,l}_{i}$. \\
\ForEach{\text{path} $\Lambda^{t,l}_{i}\normalfont{\texttt{[}j\normalfont{\texttt{]}}} \in \Lambda^{t,l}_{i}$}{%

    $\text{pathVC} \leftarrow \texttt{sort(}\texttt{set([}\texttt{node.visitCount}~\text{for}~\texttt{node}~\text{in}~\Lambda^{t,l}_{i}\texttt{[}j\texttt{]]))}$; \\
    Create an empty look-up table $\text{tb} \leftarrow \texttt{\{\}}$;\\
    \ForEach{\text{index} $k \in \normalfont{\text{pathVC}}$}{
        $\text{tb}\texttt{[}\text{pathVC}\texttt{[}k\texttt{]}\texttt{]} \leftarrow k$; \\
    }
    Create an empty vector $\text{ls} \leftarrow \texttt{[]}$;\\
    \ForEach{\text{node} $\in \text{trajectory}~\Lambda^{t,l}_{i}\normalfont{\texttt{[}j\normalfont{\texttt{]}}}$}{
        $\text{ls}\texttt{.append(} \text{tb}\texttt{[}\text{node}\texttt{.visitCount}\texttt{]} \texttt{)}$;\\
    }
    $\text{vc}\texttt{.append(}\text{ls}\texttt{)}$;
}
\textbf{Return:} $\text{vc}$.
\caption{\texttt{VisitCount}}\label{alg:create_mask}
\end{algorithm2e}

\newpage

\begin{algorithm2e}[H]
\textbf{Input:} $\Lambda^{t,l}_{i}\texttt{[}j\texttt{]}, \text{vc}, \tau_i$.\\
\textbf{Initialize:} $e^{i,j, \downarrow}_{t}\leftarrow-1$, $\text{res}\leftarrow\texttt{[]}$. \tcp{Initialize the return value and an empty vector}
$\text{vc}\leftarrow \text{vc}\texttt{[:}t\texttt{]}$; \tcp{Slicing the visit count}
$\text{left} \leftarrow 0$, $\text{right} \leftarrow 1$; \tcp{Create two pointers}
\While{$\normalfont{\text{right} < \texttt{size(} \text{vc} \texttt{)}}$}{
    \uIf{$ \normalfont{\text{vc}\texttt{[}\text{right} \texttt{]} = \text{vc}\texttt{[}\text{left} \texttt{]} } $}{
        $\text{right}\texttt{++}$; \tcp{Non-decreasing pattern}
    }
    \uElseIf{$\normalfont{ \text{vc}\texttt{[}\text{right} \texttt{]} > \text{vc}\texttt{[}\text{left} \texttt{]}  }$}{
        $\text{res}\texttt{.append(}\text{right} \texttt{)}$; \\
        $\text{left} \leftarrow \text{right}$;\\
        $\text{right} \texttt{++}$; \\
    }
    \Else{
        \texttt{break}; \tcp{Break the while loop}
    }
}
\uIf{$\normalfont{\texttt{size(} \text{res} \texttt{)} = 0}$}{
    $e^{i,j, \downarrow}_{t} \leftarrow -1$; \\
}
\Else{
    $e^{i,j, \downarrow}_{t} \leftarrow \text{res} \texttt{[} \texttt{-1]} $; \tcp{Get the last value of $\text{res}$}
}
\textbf{Return:} $e^{i,j, \downarrow}_{t}$.
\caption{\texttt{UL}}\label{alg:ul}
\end{algorithm2e}

\newpage

\begin{algorithm2e}[H]
\textbf{Input:} $\Lambda^{t,l}_{i}\texttt{[}j\texttt{]}, \text{vc}, \tau_i$.\\
\textbf{Initialize:} $e^{i,j, \uparrow}_{t}\leftarrow-1$, $\text{res}\leftarrow\texttt{[]}$. \tcp{Initialize the return value and an empty vector}
$\text{vc}\leftarrow\texttt{reverse(}\text{vc}\texttt{[:}t \texttt{])}$; \tcp{Slicing the visit count and then reverse}
$\text{left} \leftarrow 0$, $\text{right} \leftarrow 1$; \tcp{Create two pointers}
\While{$\normalfont{\text{right} < \texttt{size(} \text{vc} \texttt{)}}$}{
    \uIf{$ \normalfont{\text{vc}\texttt{[}\text{right} \texttt{]} = \text{vc}\texttt{[}\text{left} \texttt{]} } $}{
        $\text{right}\texttt{++}$; \tcp{Non-increasing pattern}
    }
    \uElseIf{$\normalfont{ \text{vc}\texttt{[}\text{right} \texttt{]} > \text{vc}\texttt{[}\text{left} \texttt{]}  }$}{
        $\text{res}\texttt{.append(}\text{right} \texttt{)}$; \\
        $\text{left} \leftarrow \text{right}$;\\
        $\text{right} \texttt{++}$; \\
    }
    \Else{
        \texttt{break}; \tcp{Break the while loop}
    }
}
\uIf{$\normalfont{\texttt{size(} \text{res} \texttt{)} = 0}$}{
    $e^{i,j, \uparrow}_{t} \leftarrow -1$; \\
}
\Else{
    $e^{i,j, \uparrow}_{t} \leftarrow \text{res} \texttt{[} \texttt{-1]} $; \tcp{Get the last value of $\text{res}$}
    $e^{i,j, \uparrow}_{t} \leftarrow \texttt{size(}\text{vc} \texttt{)} - e^{i,j, \uparrow}_{t} - 1$; \tcp{Get the right timestep as $\text{vc}$ is reversed}
}
\textbf{Return:} $e^{i,j,\uparrow}_{t}$.
\caption{\texttt{LU}}\label{alg:lu}
\end{algorithm2e}

\newpage

\begin{algorithm2e}[H]
\textbf{Input:} $\boldsymbol{e}^{i}_t$. \tcp{Receives a vector of pivot timesteps}
\tcp{Get the pivot timestep with the maximum count}
\tcp{Since $\texttt{getValueOfMaxCount(} \cdot \texttt{)}$ is easy to implement, for brevity, we do not present its implementation here}
$e^{i}_t \leftarrow \texttt{getValueOfMaxCount(} \boldsymbol{e}^{i}_t \texttt{)}$; \\
\textbf{Return:} $e^{i}_t$.
\caption{\texttt{Summarize}}\label{alg:summarize}
\end{algorithm2e}

\newpage

\begin{algorithm2e}[H]
\textbf{Input:} $\omega$, $\Lambda^{t,l}_{i}$, $\tau_i$, $\boldsymbol{r}^{l,i}$, $\Upsilon$ and $\Phi_i$.\\
\textbf{Initialize:} $e^{i}_t \leftarrow -1$; \tcp{Initialize the results}
\textbf{Initialize:} $\text{tb} \leftarrow \texttt{\{\}}$; \tcp{Initialize a empty look-up table}
\ForEach{\text{path} $\normalfont{\Lambda^{t,l}_{i}\texttt{[}j\texttt{]} \in \Lambda^{t,l}_{i}}$}{%
    \tcp{Search backwards from $t-1$ to $0$}
    \ForEach{$\normalfont{\texttt{node}~j=t-1 \in \Lambda^{t,l}_{i}\texttt{[}j\texttt{]}~\text{to}~0}$} { 
        $\text{tb}\texttt{[}j\texttt{]}\leftarrow \texttt{\{\}}$; \\
        \uIf{$\normalfont{new~\texttt{node}}$}{
            $\text{tb}\texttt{[}j\texttt{][node]=1}$;
        }
        \Else{
            $\text{tb}\texttt{[}j\texttt{][node]++}$;
        }
    }
}
\uIf{\normalfont{the size of each $\text{tb}\texttt{[}j\texttt{]}$ equals the size of $\Lambda^{t,l}_{i}$}}{
    $e^{i}_t \leftarrow -1$;\\
}
\Else{
    Get the timestep $j^{\prime}$ whose length of $\text{tb}\texttt{[}j^{\prime}\texttt{]}$ is the smallest one;\\
    $e^{i}_t \leftarrow j^{\prime}$;\\
}
\textbf{Return:} $e^{i}_t$.
\caption{Search Scheme II}\label{alg:search_scheme_II}
\end{algorithm2e}

\newpage
\subsection{\AsyncActionUpperCase MARL Training}\label{appendix_sec:marl_training}
We present the pseudo code of incorporating \OurMethod~into model-free MARL method in Alg. \ref{alg:triem_algo}. 
Lines \ref{alg:commit_actions}-\ref{alg:save_trajs} show that agents commit actions into the environment and then agents save the trajectories into the buffer. 
Agents update their individual episodic memory in line \ref{alg:call_triem}. In lines \ref{alg:start_to_train}-\ref{alg:end_update}, agents' policies are trained with TD learning (line \ref{alg:cal_td}) by searching agent's episodic memories (line \ref{alg:search_triem}). We also provide a pictorial view of our framework in Fig. \ref{fig:ctde_framework} to show the whole pipeline of MARL learning.

\begin{algorithm2e}[ht]
{
  \textbf{Input}: initialize parameters $\bar{\theta}$ and $\theta$ of the network and the target network of agents, replay buffer $\mathcal{D}$ and $\Phi$;\\
  \For{$j\leftarrow1$ {\bfseries to} max\_episode} {
    \While{\normalfont{episode\_not\_terminated}} {
        All agents commit actions $\thicktilde{\boldsymbol{u}}^{t}_{i}$ into environment; \label{alg:commit_actions} \\
        Collect $(\boldsymbol{s}^{t}, \{o^{t}_{i}\}^{N}_{i=1}, \thicktilde{\boldsymbol{u}}^{t}, r^{t}, \boldsymbol{s}^{\prime})$; save it into  $\mathcal{D}$; \label{alg:save_trajs} \\
        Call \texttt{Update\OurMethod} (Alg. \ref{alg:update_graph});  \label{alg:call_triem} \\
        \If{\normalfont{update\_the\_model}} { \label{alg:start_to_train}
            Sample a min-batch $\mathcal{D}^{\prime}$ from $\mathcal{D}$;\\
            For each sample in $\mathcal{D}^{\prime}$, get $e_t$ by calling \texttt{SearchPivotTimesteps} (Alg. \ref{alg:search_triem_algo}) ; \label{alg:search_triem}\\
            Calculate the TD target with Eqn. \ref{eq:rew_redist} and \ref{eq:new_td_loss_dqn}; \label{alg:cal_td} \\
            Update $\theta$ by minimizing the TD loss;\\
            \If{\normalfont{update}}{
                Update $\bar{\theta}$: $\bar{\theta} \leftarrow \theta$; \label{alg:end_update} \\
            }
        }
    }
  }
  \textbf{Return}: A well-trained policy for each agent.
  \caption{\AsyncActionUpperCase MARL Training} \label{alg:triem_algo}
}
\end{algorithm2e}

\subsection{List of Symbols}

\begin{table}[h]
    \centering
    \caption{List of Symbols for Dec-POMDP and \AsyncActionUpperCase Dec-POMDP}\label{tb:appendix_notation_decpomdp}
    \begin{tabular}{{l}{l}}
    \hline
    Symbol & Meaning \\
    \hline
    $\mathcal{S}$ & The state space \\
    $\mathcal{U}$ & The action space  \\
    $\mathcal{P}$ & The transition probability \\
    $R$ & The reward function \\
    $\mathcal{R}$ & The reward space \\
    $O$ & The observation function \\
    $\mathcal{O}$ & The observation space \\
    $\mathcal{N}$ & The index set agents \\
    $\bm{s}$ & The current global state \\
    $\bm{u}$ & The current action \\
    $\bm{s}'$ &  The next global state \\
    $\gamma$ &  The discount factor \\
    $i$ & The index of agent $i$ \\ 
    $u_i$ & The action of agent $i$ at current timestep \\ 
    $N$ & The number of agents \\
    $\mathcal{T}$ & The agent's action-observation-reward history space \\
    $\tau_i$ & Agent's action-observation-reward history \\
    \hline
    $\thicktilde{u}$ & Agent's \AsyncAction action \\
    $m_{\thicktilde{u}_i}$ & The execution duration of agent $i$'s action $\thicktilde{u}_i$ \\
    $A$ & The action duration distribution \\
    $\mathcal{A}$ & The space of the action duration distribution  \\
    $\bm{\thicktilde{u}}$ & The joint \AsyncAction action \\
    $\bm{\thicktilde{u}}_t$ & The joint \AsyncAction action at timestep $t$ \\
    $\bm{m}$ & The execution duration of $\bm{\thicktilde{u}}$ \\
    $\bm{m}_t$ & The execution duration of $\bm{\thicktilde{u}}_t$ \\
    \hline
    \end{tabular}
\end{table}

\newpage
\begin{table}[h]
    \centering
    \caption{List of Symbols for MARL and \AsyncActionUpperCase MARL}\label{tb:appendix_notation_marl}
    \begin{tabular}{{l}{l}}
    \hline
    Symbol & Meaning \\
    \hline
    $Q_i$ & Agent $i$'s Q value \\
    $Q^{\operatorname{tot}}$ & The global Q value of all agents \\
    $Q^{*}$ & The optimal global Q value of all agents \\
    $\theta$ & The the parameters of the agents (including agent's network, and networks for learning $Q^{\operatorname{tot}}$) \\
    $\bar{\theta}$ & The the parameter of the target network \\
    $D^{\prime}$ & A sample from the replay buffer \\
    $\mathcal{D}$ & The replay buffer \\
    \hline
    $e_t$ & The pivot timestep for $r_t$ \\
    $\Gamma$ & The \AsyncActionUpperCase Bellman operator  \\
    $\hat{r}^{e_t}$ & The redistributed reward \\
    $\Pi_{\Phi}$ & The reward redistribution operator \\
    \hline
    \end{tabular}
\end{table}

\begin{table}[h]
    \centering
    \caption{List of Symbols for \OurMethod}\label{tb:appendix_notation_ourmethod}
    \begin{tabular}{{l}{l}}
    \hline
    Symbol & Meaning \\
    \hline
    $\tau_i$ & Agent $i$'s observation-action-reward trajectory \\
    $\o^{t}_i$ & Agent $i$'s observation at timesetp $t$ \\
    $\thicktilde{u}^{t}_i$ & Agent $i$'s \AsyncAction action at timesetp $t$ \\
    $r^{t}$ & The global reward at timestep $t$ \\
    $T$ & The length of the $\tau_i$ \\
    $\Phi$ & The set of \OurMethod \\
    $\Phi_{i}$ & The set of agent $i$'s \OurMethod \\
    $\phi^{t}_{i}$ & Agent $i$'s \OurMethod whose maximum level is $t$  \\
    $\Psi$ & The set of nodes \\
    $\Xi$ & The set of edges \\
    $\omega$ & The index of episode return \\
    $\Omega$ & The length of the episode return list $\Upsilon$\\
    $\phi^{t,\omega}_{i}$ & Agent $i$'s $\omega$-th sub-graph of which maximum level is $t$ \\
    $\Phi^{t,\Omega}_{i}$ & Agent $i$'s set of sub-graphs of which maximum level is $t$ \\
    $\Upsilon$ & The episode return list \\
    $\Upsilon \texttt{[} \omega \texttt{]}$ & The $\omega$-th episode return \\
    $\boldsymbol{r}^{t,i}$ & The discretized maximum episode return of $\phi^{t}_{i}$\\
    $l$ & the index of the $l$-th level in the graph. \\
    $\Lambda^{t}_{i}$ & All the paths from node at level $t$ to node at the level $0$ \\
    $e_t$ & The pivot timestep for $r_t$ \\
    $e^{i}_t$ & Agent $i$'s pivot timestep for $r_t$ \\
    $\boldsymbol{e}^{i}_t$ & The vector of agent $i$'s pivot timestep of each path in $\Lambda^{t,l}_{i}$ for $r_t$  \\
    \hline
    \end{tabular}
\end{table}

\newpage
\section{Environments}\label{appendix:environments}

\begin{figure}[ht]
    \centering
    \includegraphics[scale=0.34]{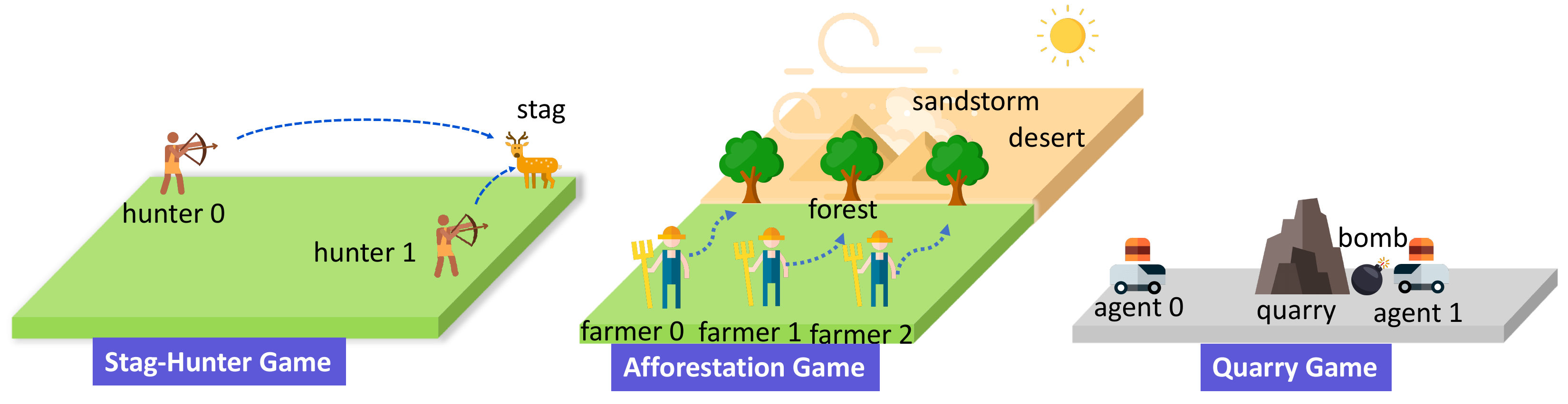}
    \caption{Stag-hunter Game, Quarry Game and Afforestation Game.}
    \label{fig:3_games_appendix}
\end{figure}

\subsection{Stag-Hunter Game} \label{appendix:stag_hunter}
As depicted in Fig. \ref{fig:3_games_appendix}, there are $n$ agents whose action durations are different; the task is to catch the stag by shooting it simultaneously for all agents. Agents cannot move and the distance between the agent and the stag is different. The stag will escape when hit by $j \in \{1, \cdots, n-1\}$ arrows. In this case, agents will receive a positive reward given the number of arrows that successfully shoot the stag. In Sec. \ref{sec:experiments}, the environment dimension of the Stag-Hunter Game is $15\times15$ and maximum time steps is 14. At each time step, each agent can only observe its position and the position of the stag. It cannot observe the position of other agents. Agents can select \texttt{SHOOT} or \texttt{NOOP} actions. \texttt{SHOOT} means shooting the arrow and \texttt{NOOP} means no actions to be executed. For agent $0$, the action duration of its \texttt{SHOOT} action is $14$ while the action duration of agent $1$ \texttt{SHOOT} action is $6$. When agent $1$ shoots the arrow at timestep $0$, all agents will receive a positive reward at the end of the episode, making it challenging for TD learning to calculate the exact contribution of each agent. 

\subsection{Quarry Game} \label{appendix:quarry}
There are $n$ agents in a quarry, as shown in Fig. \ref{fig:3_games_appendix}. Agents' task is to complete the $n$-explosive installation task, and only when all the explosives detonate will agents receive the optimal positive reward. After the installation, agents should go to the safe zones. Otherwise, agents will die and receive a negative reward when the explosive detonates. Agents will receive a medium-level reward given the number of detonated explosives. The explosive has different period to detonate after the installation.  At each time step, each agent can observe its position, the position of the quarry, the position of the explosive set by the agent (if any) and the time seconds left for the explosive set by the agent (if any).  The agent can select \texttt{MOVE\_LEFT}, \texttt{MOVE\_RIGHT}, \texttt{NOOP} or \texttt{INSTALL\_EXPLOSIVE} actions at each time step. Note that each agent cannot observe the status of the other agents and others’ explosives. Episode ends after the maximum timesteps or the explosive detonates. To complete the task, agents should place the explosive at the right timestep and return to the safe zone.

\subsection{Afforestation Game} \label{appendix:afforestation}
In Fig. \ref{fig:3_games_appendix}, there are $n$ farmer agents in the farm. To the north of the farm, there is a desert. In the early spring, strong sandstorms may gust from the north and destroy the farm. In order to protect the farm, farm agents should plant trees in the north of the farm. Only when trees are tall enough can they protect the farm. Trees can have different durations to grow, making the \AsyncActionNotilde ness of the planting action. Agents receive the optimal positive reward when there are $n$ trees can protect the farm before the sandstorm. Note that agents have partial observations and will receive a reward of $-0.1$ at each timestep in all scenarios shown in Fig. \ref{fig:3_games_appendix}. At each time step, each agent can observe its position, the position of the trees planted by itself and the status (position and the age of the tree) of the sandstorm (if any). The agent cannot observe the other agents’ positions and trees planted by other agents. Agents can take \texttt{MOVE\_NORTH}, \texttt{MOVE\_SOUTH}, \texttt{NOOP} or \texttt{PLANT\_TREE} actions. Similar to Quarry Game and Stag-Hunter, agents should take the right action at the right timestep to complete the task. At the last timestep of the episode, the standstorm will gust and if there are enough trees to protect the farm, the task will be success. Note that agents should return to the safe zone when the standstorm comes. Otherwise, agents will receive a fraction of the punishment. Agents will receive rewards when they fail to return to the safe zone. The reward is proportional to the number of agents who fail to return to the safe zone.

\subsection{SMAC}
We introduce the \AsyncAction action into SMAC~\cite{samvelyan19smac}. To make reasonable changes of the environment, we select the \texttt{ATTACK} actions as \AsyncAction actions for SMAC scenarios. The rest actions, including \texttt{MOVE}, \texttt{NOOP} will be executed immediately into the environment after inference. At each time step, agents get local observations within their field of view, which contains information (relative x, relative y, distance, health, shield, and unit type) about the map within a circular area for both allied and enemy units and makes the environment partially observable for each agent. All features, both in the global state and in individual observations of agents, are normalized by their maximum values. Actions are in the discrete space: move[direction], attack[enemy id], stop and no-op. The no-op action is the only legal action for dead agents. Agents can only move in four directions: north, south, east, or west.

\newpage
\section{Baselines} \label{appendix:baselines}

We introduce the baselines evaluated in the experimental section. All baselines are summarized in Table \ref{tab:baselines_appendix}.

\textbf{IQL}~\cite{tampuu2017multiagent}: IQL is an independent Q-learning method for multi-agent RL. Each agent learns its Q values independent with Q-learning~\cite{watkins1992q}.

\textbf{VDN}~\cite{sunehag2017value}: VDN uses a linear combination of individual Q values to approximate the $Q^{\operatorname{tot}}(\boldsymbol{\tau}, \boldsymbol{u})$ as $Q^{\operatorname{tot}} = \sum_{i=1}^{N} Q_{i}(\tau_{i}, u_{i})$.

\textbf{QMIX}~\cite{rashid2018qmix}:
QMIX introduces the monotonic constraint on the relationship between $Q^{\operatorname{tot}}$ and $Q_i$:
\begin{equation}
\frac{\partial Q^{\operatorname{tot}}(\boldsymbol{\tau}, \boldsymbol{u})}{\partial Q_{i}(\tau_{i}, u_{i})} \geq 0, \forall i \in\{1,2, \ldots, N\} \nonumber
\end{equation}
where $Q^{\operatorname{tot}(\boldsymbol{\tau}, \boldsymbol{u})}=f_{m} (Q_{1}(\tau_{1}, u_{1}), \dots, Q_{N}(\tau_{N}, u_{N}))$ and $f_{m}$ is a mixing network used to approximate the $Q^{\operatorname{tot}}$. 

\textbf{QTRAN}~\cite{son2019qtran}: QTRAN factorize the $Q_{\operatorname{jt}}(\boldsymbol{\tau}, \boldsymbol{u})$ with transformation:
\begin{equation}
\sum_{i=1}^{N} Q_{i}\left(\tau_{i}, u_{i}\right)-Q_{\mathrm{jt}}(\boldsymbol{\tau}, \boldsymbol{u})+V_{\mathrm{jt}}(\boldsymbol{\tau})=\left\{\begin{array}{ll}
0 & \boldsymbol{u}=\overline{\boldsymbol{u}} \\
\geq 0 & \boldsymbol{u} \neq \overline{\boldsymbol{u}}
\end{array}\right.\nonumber
\end{equation}
where $V_{\mathrm{jt}}(\boldsymbol{\tau})=\max _{\boldsymbol{u}} Q_{\mathrm{jt}}(\boldsymbol{\tau}, \boldsymbol{u})-\sum_{i=1}^{N} Q_{i}\left(\tau_{i}, u_{i}\right)$. 

\textbf{QPLEX}\cite{wang2020qplex}: Wang et al.~\cite{wang2020qplex} utilizes the established dueling structure
$Q = V + A$~\cite{wang2016dueling}, advantage function and attention network~\cite{vaswani2017attention} and introduces the following factorization:
\begin{equation}
\begin{array}{l}
Q_{t o t}(\boldsymbol{\tau}, \boldsymbol{u})=V_{t o t}(\boldsymbol{\tau})+A_{t o t}(\boldsymbol{\tau}, \boldsymbol{u}) \\ V_{t o t}(\boldsymbol{\tau})=\max _{\boldsymbol{u}^{\prime}} Q_{t o t}\left(\boldsymbol{\tau}, \boldsymbol{u}^{\prime}\right)\\
Q_{i}\left(\tau_{i}, a_{i}\right)=V_{i}\left(\tau_{i}\right)+A_{i}\left(\tau_{i}, a_{i}\right) \\ V_{i}\left(\tau_{i}\right)=\max _{a_{i}^{\prime}} Q_{i}\left(\tau_{i}, a_{i}^{\prime}\right)
\end{array}\nonumber
\end{equation}

\textbf{EMC}~\cite{zheng2021EMC}: EMC is MARL episodic memory method that utilizes episodic memory of RL~\cite{Zhu2020Episodic,HuHaoGEM2021} in MARL curiosity-driven exploration.


\textbf{$N$-step Return and TD($\lambda$) methods}~\cite{sutton1984temporal,sutton2018reinforcement}: $N$-step Return and TD($\lambda$) are methods for Q value prediction.

\textbf{MAVEN}~\cite{mahajan2019maven}: MAVEN builds a latent space with multual information for multi-agent exploration.

\textbf{RMIX}~\cite{qiu2021rmix}: RMIX aims to learnig risk-sensitive policies for MARL. It replaces the Q value policy with  CVaR~\cite{rockafellar2000optimization} for risk-sensitive policy learning.

\begin{table}[ht]
    \vspace{-0.5cm}
    \caption{Baseline algorithms.}
    \label{tab:baselines_appendix}
    \centering
    \begin{tabular}{crcrcr}
        \toprule
        
        \multicolumn{2}{c}{Categories} &
        \multicolumn{2}{l}{Methods} \\
        \cmidrule(lr){1-2}
        \cmidrule(lr){3-4}
        
        \multicolumn{2}{c}{\multirow{5}{*}{\makecell{MARL Baselines  (\textbf{Q1})}}} &  
        \multicolumn{2}{l}{QMIX\cite{rashid2018qmix}} \\
        \multicolumn{2}{c}{} & \multicolumn{2}{l}{VDN~\cite{sunehag2017value}} \\
        \multicolumn{2}{c}{} & \multicolumn{2}{l}{IQL~\cite{tampuu2017multiagent}} \\
        \multicolumn{2}{c}{} & \multicolumn{2}{l}{QTRAN~\cite{son2019qtran}} \\
        \multicolumn{2}{c}{} & \multicolumn{2}{l}{QPLEX~\cite{wang2020qplex}} \\
        
        \cmidrule(lr){1-2}
        \cmidrule(lr){3-4}

        \multicolumn{2}{c}{\multirow{1}{*}{\makecell{EM (\textbf{Q2})}}} & 
        \multicolumn{2}{l}{EMC~\cite{zheng2021EMC}} \\
        


        \cmidrule(lr){1-2}
        \cmidrule(lr){3-4}

        \multicolumn{2}{c}{\multirow{2}{*}{\makecell{Bootstrap  (\textbf{Q3})}}} & 
        \multicolumn{2}{l}{N-step Return \& } \\
        \multicolumn{2}{c}{} & \multicolumn{2}{l}{$\lambda$-Return~\cite{sutton2018reinforcement}} \\

        \cmidrule(lr){1-2}
        \cmidrule(lr){3-4}

        \multicolumn{2}{c}{\multirow{3}{*}{\makecell{Ex-Risk (\textbf{Q4})}}} & 
        \multicolumn{2}{l}{MAVEN~\cite{mahajan2019maven}} \\
        \multicolumn{2}{c}{} & \multicolumn{2}{l}{EMC~\cite{zheng2021EMC}} \\
        \multicolumn{2}{l}{} & \multicolumn{2}{l}{RMIX~\cite{qiu2021rmix}} \\
        \toprule
    \vspace{-0.8cm}
    \end{tabular}
\end{table}

\newpage
\section{Experiment Settings} \label{appendix:experiments}

We implement our method on PyMARL~\cite{samvelyan19smac} and use 10 random seeds to train each method on all environments. We use opensourced code of baselines publicly by the corresponding authors on Github in all experiments. We use the default settings of PyMARL in our research, including the relay buffer, the mixing network, the training hyperparameters. In order to explore, we use $\epsilon$-greedy with $\epsilon$ annealed linearly from 1.0 to 0.05 over 50K time steps from the start of training and keep it constant for the rest of the training for all methods. The discount factor $\gamma=0.99$ and we follow the default hyper-parameters used in the original papers of all methods in our research. We carry out experiments on NVIDIA A100 Tensor Core GPU and NVIDIA GeForce RTX 3090 24G. We resort to mean-std values as our performance evaluation measurement. We use $\beta=0.00001$ in Eqn.~\ref{eq:rew_redist}. To create sub-graphs in \OurMethod, we first calculate the episode return and keep 1 decimal of it. We then use this episode return to create each sub-graph. We list some important hyper-parameters in Tab.~\ref{tb:hyparam}.

\begin{table*}[ht]
\caption{Hyper-parameters}\label{tb:hyparam}
\centering
\begin{tabular}{{c}{c}}
    \hline
    hyper-parameter & Value \\
    \hline
    Optimizer & RMSProp \\
    Learning rate  & 5e-4 \\
    RMSProp alpha & 0.99\\
    RMSProp epsilon & 0.00001 \\
    Gradient norm clip & 10 \\
    \hline
    Batch size & 32 \\
    Replay buffer size &  5,000 \\
    \hline
    Exploration method  &  $\epsilon$-greedy \\
    $\epsilon$-start & 1.0 \\
    $\epsilon$-finish & 0.05 \\
    $\epsilon$-anneal time & 50,000 steps \\
    \hline
    $\gamma$ & 0.99 \\
    $\beta$ & 0.00001\\
    Evaluation interval & 10,000 \\
    Target update interval & 200 \\
    \hline
\end{tabular}
\end{table*}




\newpage
\section{Experiment Results} \label{appendix:exp_results}

We provide additional experiment results on $n$-step return and TD($\lambda$). As illustrated in Fig. \ref{fig:stag_hunter_n_steps_more} and \ref{fig:stag_hunter_lambdas_more}, QMIX can attain acceptable performance with some specific values of $n$ and $\lambda$ in Stag-Hunter Game. However, there is no convincing improvements of performance of VDN, QTRAN and IQL. On Quarry Game (Fig. \ref{fig:quarry_n_steps_more} and \ref{fig:quarry_lambdas_more}) and Afforestation Game (Fig. \ref{fig:afforest_n_steps_more} and \ref{fig:afforest_lambdas_more}), we can find that TD($\lambda$) cannot help to improve the performance of MARL methods. We can conclude that $n$-step and TD($\lambda$) have limited ability on improving the performance of MARL methods on OBMAS.


In addition to the empirical results of $n$-step and TD($\lambda$) returns, we present the results of MARL methods on Afforestation Game. In Fig. \ref{fig:afforest_results_1}, we can find that with~\OurMethod, all four methods get improved performance. We also compare QMIX-\OurMethod~and VDN-\OurMethod~with EMC, MAVEN, QPLEX, and RMIX. Despite the simple structure of VDN, VDN-\OurMethod~performs well and even outperforms QMIX-\OurMethod, demonstrating comparable performance with RMIX as depicted in Fig. \ref{baselines_1}. In Afforestation Game, agents will receive reward when they fail to return to the safe zone. The reward is proportional to the number of agents who fail to return to the safe zone. Such a clear and simple reward rule (\textit{i.e.}, a ``hint" for agents) makes learning much easier than that on Quarry and Stag-Hunter Game. This is the main reason why RMIX performs well. We can also find that MAVEN is also showing good performance due to its latent space learning model, which can efficiently learn the environment dynamics of Afforestation Game. QMIX-\OurMethod~also shows good performance and it outperforms EMC and QPLEX.

\begin{figure*}[ht]
    \centering
    \includegraphics[scale=0.375]{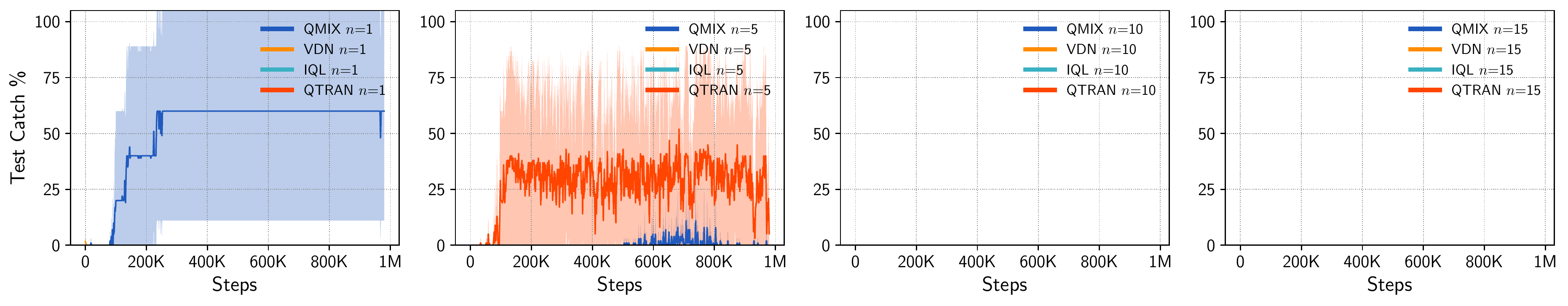}
    \caption{Results of $n$-step return Stag-Hunter Game.}
    \label{fig:stag_hunter_n_steps_more}
\end{figure*}

\begin{figure*}[ht]
    \vspace{-0.1cm}
    \centering
    \includegraphics[scale=0.375]{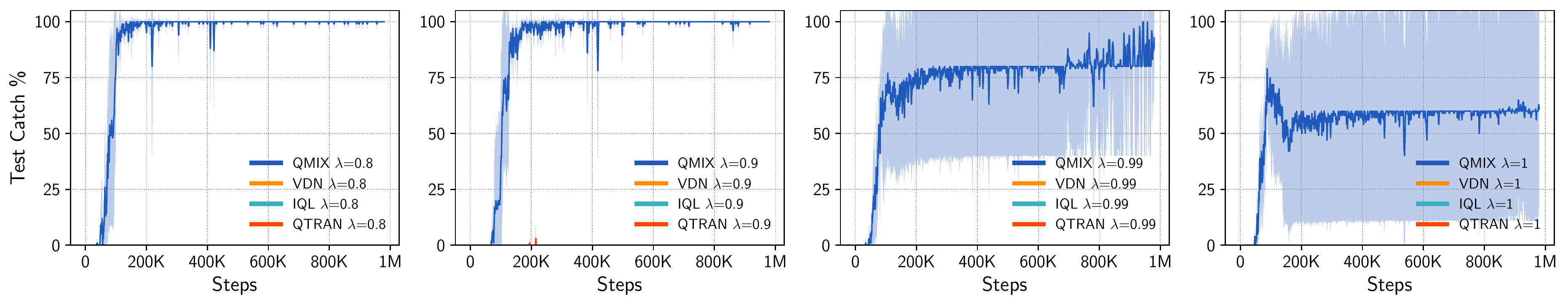}
    \caption{Results of TD($\lambda$) Stag-Hunter Game.}
    \label{fig:stag_hunter_lambdas_more}
\end{figure*}

\begin{figure*}[ht]
    \centering
    \includegraphics[scale=0.375]{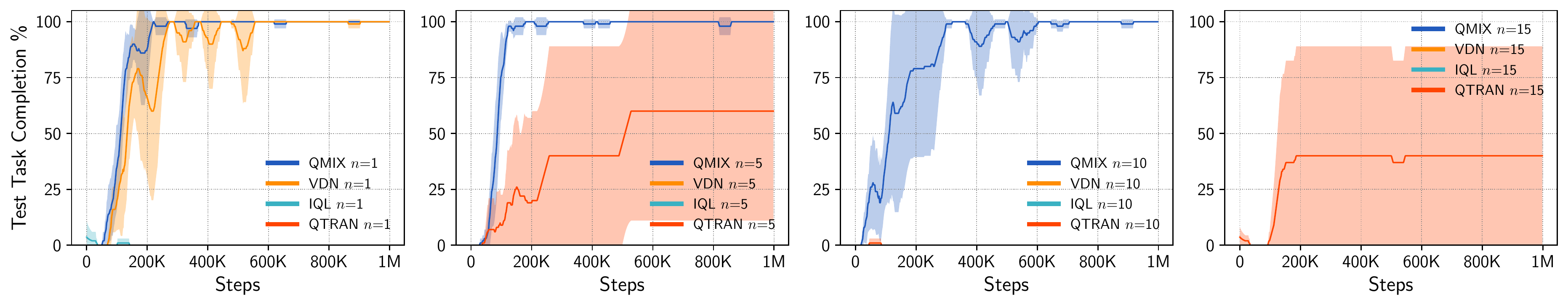}
    \caption{Results of $n$-step return Quarry Game.}
    \label{fig:quarry_n_steps_more}
\end{figure*}

\begin{figure*}[ht]
    \centering
    \includegraphics[scale=0.375]{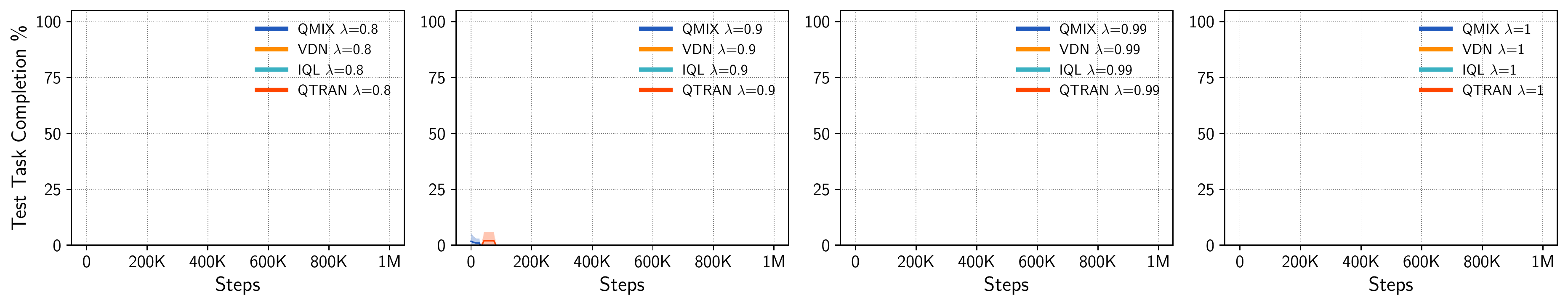}
    \caption{Results of TD($\lambda$) Quarry Game.}
    \label{fig:quarry_lambdas_more}
\end{figure*}

\begin{figure*}[ht]
    \centering
    \includegraphics[scale=0.375]{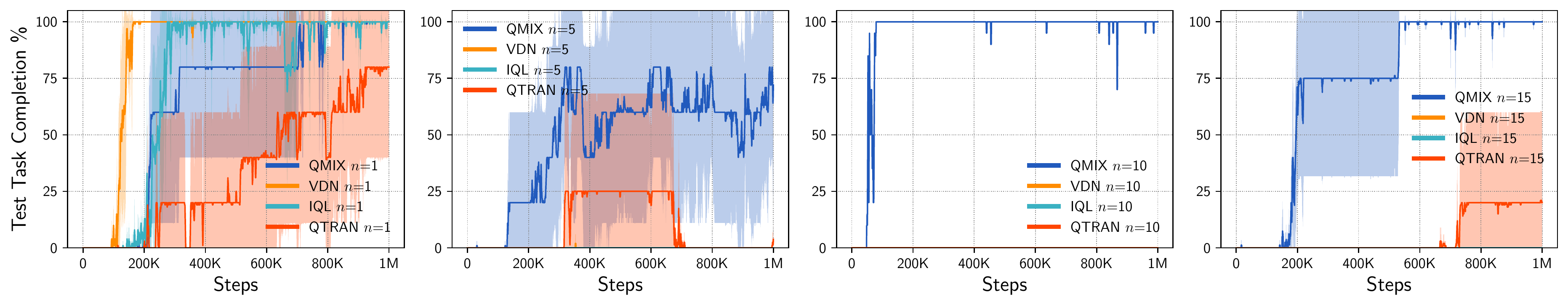}
    \caption{Results of $n$-step return Afforestation Game.}
    \label{fig:afforest_n_steps_more}
\end{figure*}

\begin{figure*}[ht]
    \centering
    \includegraphics[scale=0.375]{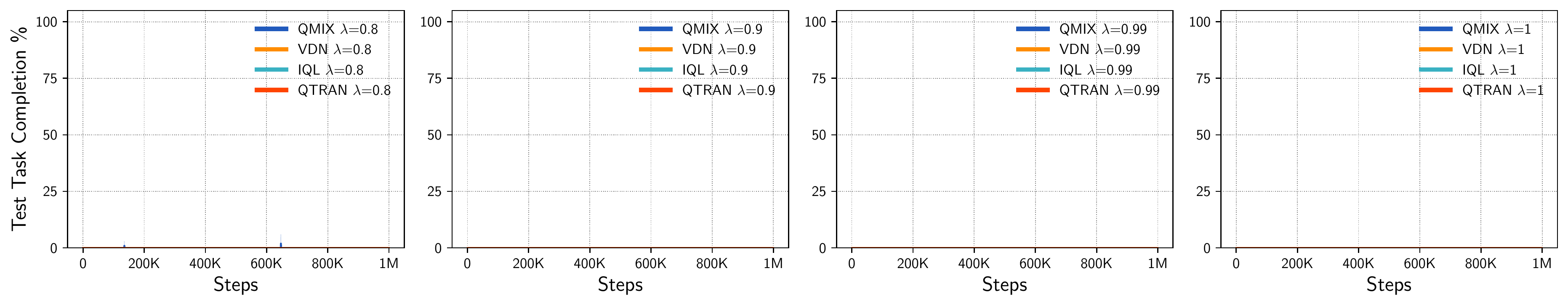}
    \caption{Results of TD($\lambda$) Afforestation Game.}
    \label{fig:afforest_lambdas_more}
\end{figure*}

\begin{figure*}[ht]
    \centering
    \includegraphics[scale=0.35]{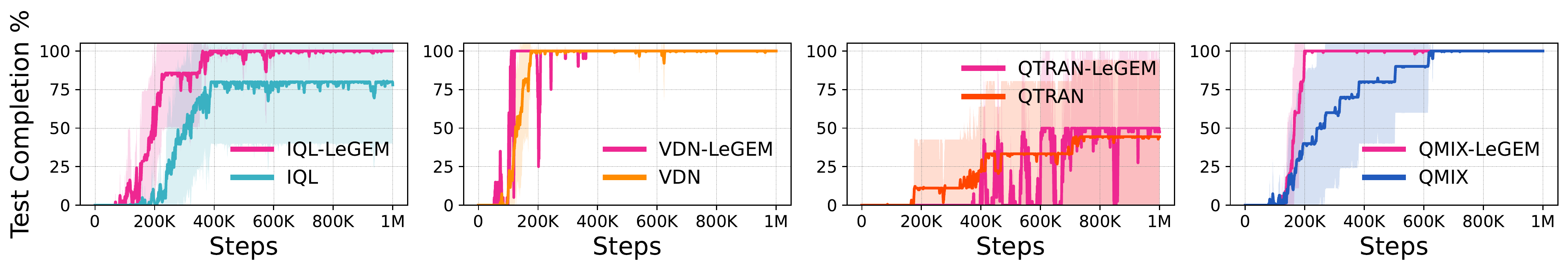}
    \caption{Results on Afforestation Game.}
    \label{fig:afforest_results_1}
\end{figure*}

\everypar{\looseness=-1}
\end{document}